\documentclass{article}

\usepackage{arxiv}

\usepackage[utf8]{inputenc} 
\usepackage[T1]{fontenc}    
\usepackage[hidelinks]{hyperref}       
\usepackage{url}            
\usepackage{booktabs}       
\usepackage{amsfonts}       
\usepackage{nicefrac}       
\usepackage{microtype}      
\usepackage{natbib}
\usepackage{bm}
\usepackage{color}
\usepackage{graphicx}
\usepackage{caption}
\usepackage{subcaption}
\usepackage{multirow}
\usepackage{bbm}
\usepackage{mathtools} 
\usepackage{amsmath}
\usepackage{amsthm}
\usepackage{mathrsfs}
\usepackage{float}
\usepackage{doi}
\usepackage{algorithm}
\usepackage{algpseudocode}
\usepackage{diagbox}
\usepackage{tikz}
\usepackage[flushleft]{threeparttable}
\usepackage{makecell}
\usepackage{scalerel}

\usetikzlibrary{decorations.markings}
\newtheorem{prop}{Proposition}
\newtheorem*{remark}{Remark}
\newtheorem{definition}{Definition}

\setcitestyle{numbers,open={[},close={]}}

\title{A new self-exciting jump-diffusion process for option pricing}
\author{
  Luis A. Souto Arias\\
  Mathematical Institute, Utrecht University\\
  The Netherlands\\
  \texttt{l.a.soutoarias@uu.nl} \\
  \And
  Pasquale Cirillo \\
  ZHAW School of Law and Management\\
  Zurich University of Applied Sciences \\
  Switzerland \\
  \texttt{ciri@zhaw.ch} \\
  \And
  Cornelis W. Oosterlee \\
 Mathematical Institute,  Utrecht University\\
  The Netherlands\\
  \texttt{c.w.oosterlee@uu.nl} \\
}

\begin{document}

\maketitle

\begin{abstract}
We propose a new jump-diffusion process, the Heston-Queue-Hawkes (HQH) model, combining the well-known Heston model and the recently introduced Queue-Hawkes (Q-Hawkes) jump process. Like the Hawkes process, the HQH model can capture the effects of self-excitation and contagion. However, since the characteristic function of the HQH process is known in closed-form, Fourier-based fast pricing algorithms, like the COS method, can be fully exploited with this model. Furthermore, we show that by using partial integrals of the characteristic function, which are also explicitly known for the HQH process, we can reduce the dimensionality of the COS method, and so its numerical complexity. Numerical results for European and Bermudan options show that the HQH model offers a wider range of volatility smiles compared to the Bates model, while its computational burden is considerably smaller than that of the Heston-Hawkes (HH) process.
\end{abstract}


\section{Introduction}

In recent years, a considerable amount of research has focused on using jumps to capture the effects of self-excitation and contagion observed in financial markets (e.g., \cite{ait-sahalia2015}). The usual Poisson process, due to the independence of its increments, is not able to capture these interrelated effects of self- and cross-excitation. While some studies advocate for jump-in-volatility models to solve this problem (see \cite{eraker2004}), many authors have opted for a more complex jump model instead. Among the proposals in the literature, the Hawkes process of \cite{Hawkes1971a,Hawkes1971b} and \cite{Hawkes1974} with exponential decay function has been extensively studied and applied in many areas of finance. Its success is due to its ability to introduce contagion effects while still being a parsimonious model, which is simple to simulate. Moreover, since it is an affine jump-diffusion (AJD) process, its characteristic function can be obtained by numerically solving a system of ODEs (see \cite{errais2010},\cite{duffie2000}). Among others, the Hawkes process has been applied to: the modelling of interest rates (\cite{sun2021},\cite{dassios2019}); switching regime models (\cite{hainaut2019}); counterparty credit risk (\cite{ma2017}); VIX option pricing (\cite{jing2021}); and to price more complex options such as power exchange options (\cite{pasricha2021}) and variance swaps (\cite{liu2019}). A comprehensive and recent survey of the role of the Hawkes process in finance can be found in \cite{hawkes2020}.

However, the lack of closed-form expressions for either the density or the characteristic function means that its estimation heavily relies on computational methods. These can become very expensive to price path-dependent options---such as Bermudan or barrier options---but even for European style options, some works still rely purely on simulations (\cite{pasricha2021},\cite{ma2017}). In spite of this, the Hawkes process has remained the model of choice for contagion, and many variants of the original process have emerged to cope with more specific and complex phenomena. It is in this context that an alternative to the Hawkes process with fully analytical characteristic function was recently introduced in \cite{Daw2022}. The process there developed, the so-called Queue-Hawkes (Q-Hawkes) process---a particular type ephemerally self-exciting (ESE) process---enjoys similar properties to the Hawkes process, such as self-excitation and reversion to the base level. The main difference between them is that, while for the Hawkes process the reversion is a function of time, for the Q-Hawkes it is fully random. These improved analytical properties make the Q-Hawkes process an interesting alternative to the original Hawkes process in financial applications. To the best of our knowledge, this process has not yet been used in financial applications.

It is, therefore, the main purpose of this paper to compare the Q-Hawkes and Hawkes jump-diffusion processes in several option pricing scenarios, ranging from simple European calls and puts to path-dependent options like Bermudan options. In particular, we model the asset dynamics as a Heston jump-diffusion process, where the diffusion component is given by the Heston model, and the jump component either by the Q-Hawkes or Hawkes processes.

Due to the availability of the characteristic function of the Q-Hawkes process, we use the COS method to price European and Bermudan options (see \cite{fang2008}, \cite{fang2011} and \cite{ruijter2012}). Notice, though, that the Q-Hawkes is a pure jump process, and so its distribution is discrete. On the other hand, the COS method, being based on the cosine transform, can only approximate continuous functions. Thus, in the present paper we propose an extension of the original COS method to cope with discrete distributions. In addition, the characteristic function of the Q-Hawkes process enjoys some integrability properties, as we show in Section \ref{sec:jump_diff}. Hence, in Section \ref{sec:cos} we also illustrate how to modify the COS method to take advantage of this properties. In this way, the proposed methodology takes full advantage of this particular model to efficiently price European and Bermudan options in the presence of self-exciting behaviour.

The paper is organized as follows. In Section \ref{sec:general}, we define the general model for the stock price in the risk-neutral measure, so that we can deploy the usual no-arbitrage arguments to price financial derivatives. In Section \ref{sec:esep}, a detailed account of the main properties of the Q-Hawkes process is given. Moreover, a brief summary of the properties of the Hawkes process is provided in Section \ref{sec:hawk} for completion. Section \ref{sec:cos} focuses on the COS method. First, we propose an extension to cope with discrete random variables and analyze its convergence. Then, we show that under some integrability assumptions the dimension of the COS method can be reduced. The last subsection describes the pricing algorithm for Bermudan options, which is based on \cite{fang2011}. Section \ref{sec:results} contains the numerical experiments, starting from a comparison of the Q-Hawkes and Hawkes processes, and then comparing the prices they yield for European and Bermudan options. Finally, Section \ref{sec:conclusions} concludes the paper.

\section{Self-exciting jump-diffusion processes}\label{sec:jump_diff}

In this section, we develop a jump-diffusion model where the jump component is represented by a self-exciting jump process. For the diffusion component, we use the well-known Heston model \cite{heston1993}. In particular, the jumps affect only the stock, so that the variance and jump processes are independent. The main part of this section is the definition of the Heston-Queue-Hawkes (HQH) process, and the derivation of its characteristic function. When the jumps are given by the Hawkes process instead, we recover the Heston-Hawkes (HH) model, that has been partially studied in \cite{ait-sahalia2015}, \cite{liu2019} and \cite{jing2021}, among others. Here we extend those studies by showing that the joint characteristic function of this model can be expressed as a system of ODEs, which need to be solved numerically. Moreover, because in this paper we use this processes for the purpose of option pricing, the dynamics are already given under the risk-neutral\footnote{If the aim of the model is, for example, time series forecasting, the dynamics in the physical measure $\mathbb{P}$ can also be defined by omitting the intensity component in the asset price of Equation \eqref{eq:esepjd_sde} and by using a different drift than the risk-free interest rate.} measure $\mathbb{Q}$.

\subsection{The general model}\label{sec:general}

Let $(\Omega,\mathcal{F},\mathbb{Q})$ be a complete probability space, equipped with a filtration $\{\mathcal{F}_t\}_{t\geq 0}$, to which all the processes defined below are adapted. Under the risk-neutral  measure $\mathbb{Q}$, the dynamics of the asset price $S_t$ are defined as
\begin{align}\label{eq:esepjd_sde}
    \frac{dS_t}{S_{t_-}} &= (r - \bar{\mu}_Y \lambda_{t_-})\, dt + \sqrt{V_t}\,dW_t^S + (e^{Y_t}-1)\, dN_t,\\
    dV_t &= \kappa(\theta-V_t)+\eta\sqrt{V_t}dW_t^V,
\end{align}
where $V_t$ is the variance process; $r$ is the risk-free interest rate; $W_t^S$ and $W_t^V$ are correlated standard Brownian motions with correlation\footnote{Since the variance and jump intensity processes are independent, the moments of the log-asset price may become infinite depending on the parameter set. This is a consequence of the diffusion part being driven by the Heston model \cite{oosterlee2019}.} $\rho$; $\kappa$, $\theta$ and $\eta$ are the variance speed of mean reversion, long-term mean and vol-of-vol parameters, respectively; and $N_t$ is a counting process with stochastic intensity $\lambda_t$. Finally, the $Y_t$ are i.i.d random variables with support in $(-\infty,\infty)$ and $\bar{\mu}_Y = \mathbb{E}[e^Y-1]$, where $Y$ is referred to as the log-jump size. In this article, we consider two functional forms for the jump intensity $\lambda_t$. Namely, the intensity of the Q-Hawkes process
\begin{equation}\label{eq:qh_inten}
    d\lambda_t = \alpha(dN_t - dN_t^{Q}), 
\end{equation}
and the intensity of the Hawkes process
\begin{equation}\label{eq:h_inten}
    d\lambda_t = \beta(\lambda^*-\lambda) + \alpha dN_t,
\end{equation}
where $N_t$ and $N_t^Q$ are counting processes with stochastic intensities $\lambda_t$ and $\beta Q_t$, respectively; and $\alpha$, $\beta$ and $\lambda^*$ represent the clustering rate, expiration rate and baseline intensity, respectively. A detailed description of these models is to be found in Sections \ref{sec:esep} and \ref{sec:hawk}.

Going back to Equation \eqref{eq:esepjd_sde}, notice that, in the absence of jumps, this reduces to the classical Heston model. When jumps occur, but $\alpha=0$, the jump intensity $\lambda_t$ is constant, and we recover the Bates model \cite{bates1996}. Finally, for $\alpha \geq 0$ the jumps tend to cluster due to the increase in the intensity, generating self-excitation. The way in which this phenomenon manifests itself depends on the particular dynamics of the jump processes, given by Equations \eqref{eq:qh_inten} and \eqref{eq:h_inten}, respectively.

As a result of the independence between the jump and diffusion components, this model presents a high level of tractability. For example, the following proposition gives an analytic expression for the asset dynamics.
\begin{prop}\label{prop:exact_dyn}
    Let
    \begin{equation*}
        D_t = S_0 \exp\left(rt - \frac{1}{2}\int_0^t V_u du + \int_0^t \sqrt{V_u}\,W_u^S - \bar{\mu}_Y\int_0^t \lambda_{u_-} du\right),
    \end{equation*}
    and
    \begin{equation*}
        J_t = \prod_{j=1}^{N_t}e^{Y_j}.
    \end{equation*}
    Then
    \begin{equation}\label{eq:esepjd_an}
        S_t = D_t\,J_t
    \end{equation}
\end{prop}
\begin{proof}
The proof is derived directly from an application of Ito's lemma for jump-diffusion processes. A similar proof can be found in \cite{pasricha2021} for the HH process.
\end{proof}

Notice that, in Proposition \ref{prop:exact_dyn}, the contributions of the variance and the jumps are completely separated. Thus. we can sample the variance as we would do for the Heston model, and simulate the jumps separately. This can be done efficiently using the thinning algorithm of \cite{ogata1981} for point processes. An application of this algorithm to the specific case of the Hawkes process can be found in \cite{pasricha2021}. The thinning algorithm for the particular case of the Q-Hawkes process is presented in Algorithm \ref{alg:thin_esep}.
\begin{algorithm}
\caption{Thinning algorithm for the Q-Hawkes process: pseudocode.}\label{alg:thin_esep}
\begin{algorithmic}
\State Set $s = 0$, $\lambda = \lambda_0$, $n = 0$, $n^{\beta} = 0$.
\While{$s\leq t$}
\State Generate a random variable $U \sim U(0,1)$. \Comment{$U(0,1)$ denotes a standard uniform distribution.}
\State Set $s = s - \ln(U)/(\lambda+\frac{\beta}{\alpha}(\lambda-\lambda^*))$. If $s>t$, exit the loop and end the program. Otherwise a new event occurs, but we still have to determine if it is a $T$ or $T^{\beta}$ jump time.
\State Generate $V \sim U(0,1)$. If $V \leq \lambda/(\lambda+\frac{\beta}{\alpha}(\lambda-\lambda^*))$, set $n = n + 1$ and $T_n = s$. Otherwise set $n^{\beta} = n^{\beta} + 1$ and $T_n^{\beta} = s$.
\State Set $\lambda = \lambda_0 + \alpha \,(n-n^{\beta})$.
\EndWhile
\end{algorithmic}
\end{algorithm}

Whereas combining Algorithm \ref{alg:thin_esep} with an efficient sampling scheme for the Heston model already provides an efficient way of simulating the asset prices, the characteristic function of the log-asset price $X_t = \log(S_t)$ is known in closed-form when the dynamics of the jump process are given by the Q-Hawkes process. On the other hand, a semi-analytical expression for the characteristic function is available when we use the Hawkes process instead. This will allow us to make use of the COS method for fast-pricing of European and Bermudan options \cite{fang2008}. 

First note that
\begin{align}\label{eq:hestonjd_cf_def}
    \mathbb{E}[e^{i v X_t}|\mathcal{F}_0] &= \mathbb{E}\left[exp\left(i v\left(X_0+rt - \frac{1}{2}\int_0^t V_u du + \int_0^t \sqrt{V_u}\,dW_u^S\right)\right)|\mathcal{F}_0\right]\mathbb{E}\left[e^{iv M_t}|\mathcal{F}_0\right] \nonumber \\
    &= \psi_{Heston}(v,t,0)\boldsymbol{\cdot} \psi_{M}(v,t,0),
\end{align}
where we have defined
\begin{equation}\label{eq:mt}
    M_t \coloneqq \sum_{j=1}^{N_t}Y_j - \bar{\mu}_Y \int_0^t \lambda_{s_-}ds.
\end{equation}
In the above, $\psi(\cdot,t,s)$ is the characteristic function operator at time $t$ given information at time $s\leq t$, and $\psi_{Heston}(\cdot,t,s)$ is the characteristic function of the pure Heston model given the filtration $\mathcal{F}_s$.

Therefore, the characteristic function of the log-asset price can be factorized as the characteristic function of the Heston model times the characteristic function of the compensated jump term $M_t$.
\begin{remark}
Notice that Equation \eqref{eq:hestonjd_cf_def} is valid for the Q-Hawkes process, the Hawkes process, or even the Poisson process with constant intensity. Thus, for all these models we can compute the characteristic function of the jump term separately, and combine it with the diffusion term at a later stage.
\end{remark}
The remainder of this section is devoted to studying the properties of the HQH and HH models and their intensities, starting with the Q-Hawkes process.

\subsection{The Heston-Queue-Hawkes model}\label{sec:esep}

In this section, the Q-Hawkes process is discussed, expanding on the analytical properties derived in \cite{Daw2022}. Specifically, we show that the intensity of the Q-Hawkes process admits a closed-form probability mass function (PMF), and that the characteristic function of the compensated jump term $M_t$ (see Equation \eqref{eq:mt}) has an analytical solution.

\subsubsection{The Q-Hawkes process}

The Q-Hawkes process was recently introduced in \cite{Daw2022} as a possible alternative to the Hawkes process. Instead of being determined by a counting process and a time-dependent memory kernel, the Q-Hawkes process consists of two counting processses $(N_t$ and $N_t^{Q})$. The former is the analogous of the Hawkes counting process, while the latter plays the role of a stochastic memory kernel. Both of these counting processes are pure jump processes with stochastic intensities $\lambda_t$ and $\beta Q_t$, respectively, where $\beta \in \mathbb{R}^+\backslash\{0\}$. A crucial part of the model is how these two intensities are related. Following \cite{Daw2022}, we refer to $\lambda_t$ as the intensity of the Q-Hawkes process, and to $Q_t$ as the activation number. These two quantities are connected via the following formulas:
\begin{equation}\label{eq:esep_inten_diff}
    d\lambda_t = \alpha(dN_t-dN_t^{Q}),
\end{equation}
and
\begin{equation}\label{eq:esep_inten}
    \lambda_t = \lambda^* + \alpha Q_t,
\end{equation}
where $\lambda^* \in \mathbb{R}^+$ is a baseline level for $\lambda_t$ and $\alpha \in \mathbb{R}^+$ is a positive constant determining the jump size. 

Equation \eqref{eq:esep_inten} implies that $\lambda_t$ is an affine function of $Q_t$. This reduces the dimensionality of the model to two dimensions (one intensity process and one counting process) and  it is partially responsible for the analytical properties of this model. Other more complex relationships for $\lambda_t$ and $Q_t$ may not admit closed-form solutions.

Equation \eqref{eq:esep_inten_diff} shows many similarities to the intensity of the Hawkes process (see Equations \eqref{eq:hawkes_def} and \eqref{eq:hawkes_diff} in Section \ref{sec:hawk}). In both cases, the intensity $\lambda_t$ has a baseline level $\lambda^*$. Every time $N_t$ jumps, the intensity increases by a positive amount $\alpha$, also known as the clustering rate. This form of self-excitation is the same in the Hawkes model. Moreover, there is a ``loss of memory'' mechanism that erases the self-exciting effect of those jumps. The speed at which self-excitation is forgotten is given in both models by the parameter $\beta$, which is thus called the expiration rate. The main difference is that, for the Q-Hawkes process, this mechanism is a counting process ($N_t^Q$), while for the Hawkes process it is a deterministic function of time. As a result, in the Hawkes process information is continuously deleted, while in the Q-Hawkes process all information from a particular jump of $N_t$ is erased at once when $N_t^{Q}$ jumps. Thus, the analogy with a fully stochastic Hawkes process with exponential memory kernels, and the reason why the Q-Hawkes process is an ephemerally self-exciting process\footnote{In order to associate the Q-Hawkes process with other kernels of the Hawkes process, it would be necessary to use a different distribution for the jumps of $N_t^{Q}$. This is partially explored in \cite{Daw2022}, but in this article we will not pursue this line of research.}. 

From a more quantitative point of view, it can be shown that these models have the same expectations, both in intensities and in counting processes. However, the Q-Hawkes process has larger moments for any higher order $k\geq 2$ (see \cite{Daw2022}). Thus, it is expected that extreme contagion events are more probable in the Q-Hawkes process than in the Hawkes process. In principle, that would make the Q-Hawkes process more adequate to model extreme phenomena such as financial crises, where contagion is at its highest.

Nevertheless, in this paper we are interested in the application of the Q-Hawkes process to option pricing. To this end, we make use of the closed-form formula for its characteristic function, which we present in the next proposition.
\begin{prop} \label{CFQH}
    Let $Q_t$ be the activation number of a Q-Hawkes process with clustering rate $\alpha \in \mathbb{R}^+$, expiration rate $\beta \in \mathbb{R}^+\backslash\{0\}$ and baseline intensity $\lambda^* \in \mathbb{R}^+$. Then, its joint characteristic function with the counting process $N_t$ is given by
    \begin{align}\label{eq:esep_char_fun}
    \mathbb{E}[e^{i(u Q_t + v (N_t-N_s)}|\mathcal{F}_s] =& \, e^{\frac{\lambda^* \tau}{2\alpha}(\beta-\alpha-f(v))}\left(\frac{2f(v)}{f(v)+g(u,v)+e^{-\tau f(v)}(f(v)-g(u,v))} \right)^{\frac{\lambda^*}{\alpha}} \nonumber\\
    & \boldsymbol{\cdot} \left(\frac{(1-e^{-\tau f(v)})(2\beta-e^{iu}(\beta+\alpha))+e^{iu}f(v)(1+e^{-\tau f(v)})}{f(v)+g(u,v)+e^{-\tau f(v)}(f(v)-g(u,v))}\right)^{Q_s},
\end{align}
with $\tau=t-s$, $f(v) = \sqrt{(\beta+\alpha)^2-4\alpha\beta e^{i v}}$ and $g(u,v) = \beta + \alpha(1-2e^{i(u+v)})$.
\end{prop}
\begin{proof}
The joint characteristic function of $Q_t$ and $N_t^{Q}$ was developed in \cite{Daw2022}. Given Equation \eqref{eq:esep_inten}, it is simple to obtain the characteristic function for $Q_t$ and $N_t$ instead. In particular, denoting by $\psi_{Q,N}(u,v,t,s) \coloneqq \mathbb{E}[e^{i(u Q_t + v (N_t-N_s)}|\mathcal{F}_s]$ the characteristic function of the processes $(Q_t,N_t)$ at time $t$, given information at time $s \leq t$, we deduce from \cite{Daw2022} that $\psi_{Q,N}(u,v,t,s)$ is a solution to the following partial differential equation (PDE):
\begin{equation}\label{eq:Qt_cfpde}
    \frac{\partial \psi}{\partial t} -i  \left(\alpha(1 - e^{i(u+v)}) + \beta(1-e^{- iu}) \right)\frac{\partial \psi}{\partial u} = -\lambda^*(1 - e^{i(u+v)})\psi,
\end{equation}
with initial condition $\psi_{Q,N}(u,v,s,s) = e^{i u Q_s}$. One can then verify that Equation \eqref{eq:esep_char_fun} is indeed a solution to Equation \eqref{eq:Qt_cfpde}.
\end{proof}

\begin{remark}
Notice that, in Proposition \ref{CFQH}, we have chosen to present the results in terms of the activation number $Q_t$, instead of the intensity $\lambda_t$. This is because the activation number lives on the positive integers, making derivations easier. Obviously, due to Equation \eqref{eq:esep_inten}, moving from one process to the other is straightforward.
\end{remark}

\begin{remark}
Notice that, using Equation \eqref{eq:esep_inten}, we can trivially obtain the characteristic function of the pairs $(\lambda_t,Q_t)$ and $(Q_t,N_t^Q)$. Thus, Equation \eqref{eq:esep_char_fun} contains all the information about the Q-Hawkes process.
\end{remark}

The characteristic function plays a fundamental role in Fourier fast-pricing algorithms like the COS method. Hence, an analytic solution like Equation \eqref{eq:esep_char_fun} makes the Q-Hawkes process very suitable for this type of algorithms. Moreover, while it is not possible to integrate Equation \eqref{eq:esep_char_fun} with respect to $v$, we can integrate it with respect to $u$. Therefore, by setting $v=0$ in Equation \eqref{eq:esep_char_fun}, we are also able to obtain the PMF of the activation number.
\begin{prop}\label{prop:qt_dens}
    Let $Q_t$ be the activation number of a Q-Hawkes process with clustering rate $\alpha \in \mathbb{R}^+$, expiration rate $\beta \in \mathbb{R}^+\backslash\{0\}$, baseline intensity $\lambda^* \in \mathbb{R}^+$ and characteristic function given by Equation \eqref{eq:esep_char_fun}. Then, the PMF of the activation number $Q_t$ is represented by the following formula
    \begin{equation}\label{eq:ese_dens}
        \mathbb{P}[Q_t = x|Q_s] = \sum_{k=0}^x \binom{x-k+\frac{\lambda_s}{\alpha}-1}{x-k}\binom{Q_s}{k}p(\tau)^{\frac{\lambda_s}{\alpha}}(1-p(\tau))^{x-k}g(\tau)^{Q_s-k}(1-g(\tau))^k,
    \end{equation}
    with $\tau = t-s$, $p(t) = \frac{\beta-\alpha}{\beta-\alpha e^{(\alpha-\beta)t}}$ and $g(t) = \frac{\beta(1-e^{(\alpha-\beta)t})}{\beta-\alpha}$.
\end{prop}
\begin{proof}
First, we notice that, for $v=0$, Equation \eqref{eq:esep_char_fun} can be written as
\begin{equation}\label{eq:ese_charq}
    \mathbb{E}[e^{i u Q_t}|Q_s] = \left( \frac{p(\tau)}{1-(1-p(\tau))e^{iu}} \right)^{\frac{\lambda_s}{\alpha}}
    (e^{i u}(1-g(\tau))+g(\tau))^{Q_s}.
\end{equation}
The first term at the right-hand side of Equation \eqref{eq:ese_charq} is the characteristic function of a negative binomial distribution for the number of failures, given $\lambda_0/\alpha$ successes with probability $p(\tau)$, i.e.,
\begin{equation}\label{eq:neg_bin}
    NB\left(k;\frac{\lambda_0}{\alpha},p(\tau)\right) = \binom{k+\frac{\lambda_0}{\alpha}-1}{k}p(\tau)^{\frac{\lambda_0}{\alpha}} (1-p(\tau))^k.
\end{equation}

For the second term on the right-hand side of Equation \eqref{eq:ese_charq}, notice that $Q_s$ is a positive integer, so we can apply the binomial expansion. Now, the inverse Fourier transform becomes
\begin{equation}\label{eq:dirac_int}
\frac{1}{2\pi}\int_{-\infty}^{+\infty} \sum_{k=0}^{Q_s} \left(\binom{Q_s}{k} e^{i u k} g(\tau)^{Q_s-k}(1-g(\tau))^k\right)e^{-i u x}du = \binom{Q_s}{x} g(\tau)^{Q_s-x}(1-g(\tau))^x.
\end{equation}
Notice that, in Equation \eqref{eq:dirac_int}, we have used the definition of the Dirac delta function, which reads
\begin{equation}\label{eq:dirac_def}
    \delta(x-a) = \frac{1}{2\pi}\int_{-\infty}^{+\infty}e^{i u (x-a)}du.
\end{equation}
The final step is to apply the convolution theorem of the Fourier transform to Equations \eqref{eq:neg_bin} and \eqref{eq:dirac_int}, which immediately yields Equation \eqref{eq:ese_dens}.
\end{proof}

\begin{remark}
In the proof of Proposition \ref{prop:qt_dens}, we have used a known result about the negative binomial distribution. However, a detailed explanation on how to solve the complex integral can be necessary to follow the proof of Proposition \ref{prop:int_cf_mt} below. This is provided in Appendix \ref{app:comp_int} for completeness.
\end{remark}

\subsubsection{Characteristic function of the HQH model}

Since the characteristic function for the Heston model is known, it only remains to compute the last expectation in Equation \eqref{eq:hestonjd_cf_def} to obtain the characteristic function of the log-asset price. This is done in the following proposition.
\begin{prop}\label{prop:mt_char_ese}
    Let $Q_t$ be the activation number of a Q-Hawkes process with clustering rate $\alpha \in \mathbb{R}^+$, expiration rate $\beta \in \mathbb{R}^+\backslash\{0\}$, and baseline intensity $\lambda^* \in \mathbb{R}^+$. Moreover, let $Y$ be a random variable with support in $(-\infty,\infty)$ and characteristic function $\psi_Y(\cdot)$. We define the quantities
    \begin{equation*}
        f(v) = \sqrt{(\beta+\alpha(1+i v \bar{\mu}_Y))^2-4 \alpha \beta \psi_Y(v)},
    \end{equation*}
    and
    \begin{equation*}
        g(u,v) = \beta+\alpha(1+i v\bar{\mu}_Y - 2\psi_Y(v)e^{i u}).
    \end{equation*}
    Then, the characteristic function of $M_t$ and $Q_t$ is given by
    \begin{align}\label{eq:mt_char_ese}
         \mathbb{E}[e^{i u Q_t + i v (M_t-M_s)}|\mathcal{F}_s] =\, & \scalebox{1.2}{$e^{ \frac{\lambda^* \tau}{2 \alpha}(\beta-\alpha-i \alpha\bar{\mu}_Y v - f(v) )}$} \nonumber\\
        &\boldsymbol{\cdot} \left(\frac{2f(v)}{f(v)+g(u,v)+e^{-\tau f(v)}(f(v)-g(u,v))}\right)^{\frac{\lambda^*}{\alpha}} \nonumber \\
        &\boldsymbol{\cdot} \left(\frac{(1-e^{-\tau f(v)})(2\beta-e^{iu}(\beta+\alpha(1+i v\bar{\mu}_Y)))+e^{iu}f(v)(1+e^{-\tau f(v)})}{f(v)+g(u,v)+e^{-\tau f(v)}(f(v)-g(u,v))}\right)^{Q_s},
    \end{align}
    with $\tau = t-s$.
\end{prop}

\begin{proof}
First, we need the SDEs for the quantities of interest. The SDE for $Q_t$ is obtained from Equations \eqref{eq:esep_inten} and \eqref{eq:esep_inten_diff}. The SDE for $M_t$ is derived by direct differentiation, yielding
\begin{equation}
    dM_t = Y_t dN_t - \lambda_{t_-} dt.
\end{equation}
Next, it is not difficult to show, by similar arguments to those of \cite[Chapter~4]{mahmoud08}, that the characteristic function, $\psi_{Q,M}(u,v,t,s) \coloneqq \mathbb{E}[e^{iu Q_t + iv (M_t-M_s)}|\mathcal{F}_s]$, satisfies the following PDE
\begin{equation}\label{eq:esepjd_cfpde}
    \frac{\partial \psi}{\partial t} -i  \left(\alpha(1 - e^{iu}\psi_Y(v)+v \bar{\mu}_Y) + \beta(1-e^{-iu}) \right) \frac{\partial \psi}{\partial u} = -\lambda^*(1 - e^{iu}\psi_Y(v)+ v \bar{\mu}_Y)\psi,
\end{equation}
with initial condition $\psi_{Q,M}(u,v,s,s) = e^{i u Q_s}$.

We can solve Equation \eqref{eq:esepjd_cfpde} using the method of characteristics. Doing so gives us the following system of ODEs:
\begin{align*}
    \frac{\partial t(r,w)}{\partial w} &= 1, &t(r,0) &= s,\\
    \frac{\partial u(r,w)}{\partial w} &= -i(\alpha(1 - e^{iu}\psi_Y(v)+v \bar{\mu}_Y) + \beta(1-e^{-iu})),  &u(r,0) &= r,\\
    \frac{\partial \phi(r,w)}{\partial w} &= -\lambda^*(1 - e^{iu}\psi_Y(v)+ v \bar{\mu}_Y)\phi, &\phi(r,0) &= e^{i r Q_s}.
\end{align*}
This system can be solved analytically, and its solution gives the desired result of Equation \eqref{eq:mt_char_ese}. It can be verified by noticing that Equation \eqref{eq:mt_char_ese} is a solution to Equation \eqref{eq:esepjd_cfpde}.
\end{proof}

Subsequently, from Equation \eqref{eq:mt_char_ese} we obtain $\psi_{M}(v,t,s) = \psi_{Q,M}(0,v,t,s)$, which we can plug in Equation \eqref{eq:hestonjd_cf_def} to find the characteristic function of the log-asset price. In a similar way, it is possible to derive the joint characteristic function of the triplet $(X_t,V_t,Q_t)$. We can then use this function to price Bermudan options following the procedure from \cite{ruijter2012}. 

In Section \ref{sec:cos}, we shall see that, if available, some partial integrals of the characteristic function can be used to reduce the computational complexity of the COS method. Hence, the following result will prove useful in the rest of the paper.
\begin{prop}\label{prop:int_cf_mt}
    Let $Q_t$ be the activation number of a Q-Hawkes process with clustering rate $\alpha \in \mathbb{R}^+$, expiration rate $\beta \in \mathbb{R}^+\backslash\{0\}$, and baseline intensity $\lambda^* \in \mathbb{R}^+$. Let $M_t$ be the compensated jump term defined in Equation \eqref{eq:mt}, whose joint  characteristic function $\psi_{Q,M}(u,v,t,s)$ with $Q_t$ is given in Equation \eqref{eq:mt_char_ese}. The inverse Fourier transform of $\psi_{Q,M}(u,v,t,s)$ with respect to $u$ is given by
    \begin{equation}\label{eq:mt_char_int}
    \frac{1}{2\pi} \int_{-\infty}^{+\infty} \psi_{Q,M}(u,v,t,s) e^{-i u x} du = \sum_{k=0}^x I_1(x-k)I_2(k),
    \end{equation}
    where $I_1(\cdot)$ and $I_2(\cdot)$ are defined as 
    \begin{equation}\label{eq:It1}
    I_1(x) = \exp\left(\frac{\lambda^* \tau}{2 \alpha}(\beta-\alpha-i \alpha\bar{\mu}_Y v - f(v) )\right)\frac{(2 f(v))^{\scaleobj{1.1}{\frac{\lambda^*}{\alpha}}}}{h(v,\tau)^{\scaleobj{1.1}{\frac{\lambda^*}{\alpha}+Q_s}}} \binom{x+\frac{\lambda^*}{\alpha}+Q_s-1}{x}(1-p(v,\tau))^{x},
    \end{equation}
    and
    \begin{equation}\label{eq:It2}
    I_2(x) = \binom{Q_s}{x}(\hat{h}(v,\tau))^x(2\beta(1-e^{-\tau f(v)}))^{Q_s-x},
    \end{equation}
    with $\tau = t-s$, and with $f(v)$ defined as in Proposition \ref{prop:mt_char_ese}. Further,
    \begin{equation}
        h(v,\tau) = (1+e^{-t f(v)})f(v) + (1-e^{-t f(v)})(\beta+\alpha(1+i v \bar{\mu}_Y)),
    \end{equation}
    \begin{equation}
        \hat{h}(v,\tau) = (1+e^{-t f(v)})f(v) - (1-e^{-t f(v)})(\beta+\alpha(1+i v \bar{\mu}_Y)),
    \end{equation}    
    and
    \begin{equation}
        p(v,\tau) = \frac{1}{h(v,\tau)}\left((1+e^{-t f(v)})f(v) + (1-e^{-t f(v)})(\beta+\alpha(1+i v \bar{\mu}_Y-2\psi_Y(v))\right).
    \end{equation}
\end{prop}
\begin{proof}
The procedure is similar to the proof of Proposition \ref{prop:qt_dens}. Using the convolution property of the Fourier transform, we can solve the following integrals separately
\begin{align}\label{eq:It1_def}
    I_1(x) \coloneqq & \, \exp\left(\frac{\lambda^* \tau}{2 \alpha}(\beta-\alpha-i \alpha\bar{\mu}_Y v - f(v) )\right)(2f(v))^{\scaleobj{1.1}{\frac{\lambda^*}{\alpha}}} \nonumber \\ 
    & \boldsymbol{\cdot} \frac{1}{2\pi}\int_{-\infty}^{+\infty} \frac{e^{-i u x}}{(f(v)+g(u,v)+e^{-\tau f(v)}(f(v)-g(u,v)))^{\scaleobj{1.1}{\frac{\lambda^*}{\alpha}+Q_s}}} du,
\end{align}
and 
\begin{equation}\label{eq:It2_def}
    I_2(x) \coloneqq \frac{1}{2\pi}\int_{-\infty}^{+\infty} \left(\hat{h}(v,\tau)e^{iu}+2\beta(1-e^{-\tau f(v)})\right)^{Q_s} e^{-i u x} du.
\end{equation}
Notice that we have rearranged Equation \eqref{eq:mt_char_ese}, so that the denominator appears only in Equation \eqref{eq:It1_def}. Then, Equation \eqref{eq:It2_def} can be integrated, resulting in Equation \eqref{eq:It2}.

The remaining step is to solve the integral in Equation \eqref{eq:It1_def}, and show that it equals Equation \eqref{eq:It1}. For that purpose, Equation \eqref{eq:It1_def} is rearranged in a more convenient way, leading to
\begin{align}\label{eq:It1_v2}
    I_1(x) = &\, \exp\left(\frac{\lambda^* \tau}{2 \alpha}(\beta-\alpha-i \alpha\bar{\mu}_Y v - f(v) )\right)\frac{(2 f(v))^{\scaleobj{1.2}{\frac{\lambda^*}{\alpha}}}}{h(v,\tau)^{\scaleobj{1.1}{\frac{\lambda^*}{\alpha}+Q_s}}} \nonumber \\ & \boldsymbol{\cdot} \frac{1}{2\pi}\int_{-\infty}^{+\infty} \left(\frac{1}{1-(1-p(v,\tau))e^{iu}}\right)^{\scaleobj{1.1}{\frac{\lambda^*}{\alpha}+Q_s}}e^{-i ux} du,
\end{align}
The integral in Equation \eqref{eq:It1_v2} resembles Equation \eqref{eq:comp_integral1} in Appendix \ref{app:comp_int}, which is the characteristic function of a negative binomial distribution. However, in this case $p(v,\tau)$ is a complex number\footnote{The reader is referred to Appendix \ref{app:comp_int} for a detailed derivation of the inverse Fourier transform of the characteristic function of a negative binomial distribution with a real probability parameter.}, so we must show that the integral in Equation \eqref{eq:It1_v2} is indeed a negative binomial distribution with a complex-valued ``probability'' parameter.

In Figure \ref{fig:contour_integral1} of Appendix \ref{app:comp_int}, we present an appropriate contour of integration to solve Equation \eqref{eq:It1_v2} when the divergence point $u^* = i \log(1-p(v,\tau))$ is negative and purely imaginary. In this case, since the divergence point $u^*$ is a complex number, an appropriate contour of integration is given in Figure \ref{fig:contour_integral2}, where we have assumed--without loss of generality--that $\operatorname{Re}(u^*)>0$.
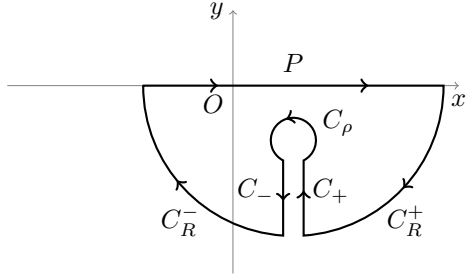
\begin{figure}[ht]
\centering
\begin{tikzpicture}
[decoration={markings,
mark=at position 1cm with {\arrow[line width=1pt]{>}},
mark=at position 3cm with {\arrow[line width=1pt]{>}},
mark=at position 5.5cm with {\arrow[line width=1pt]{>}},
mark=at position 7.6cm with {\arrow[line width=1pt]{>}},
mark=at position 8.9cm with {\arrow[line width=1pt]{>}},
mark=at position 10.14cm with {\arrow[line width=1pt]{>}},
mark=at position 12.24cm with {\arrow[line width=1pt]{>}}
}]
\draw[help lines,->] (-3,0) -- (3,0) coordinate (xaxis);
\draw[help lines,->] (0,-2.5) -- (0,1) coordinate (yaxis);

\path[draw,line width=0.8pt,postaction=decorate] (-1.2,0) -- (2.8,0) arc(0:-86:2) -- +(0,1) arc(-63:243:0.3) -- +(0,-1) arc(-94:-180:2);

\node[below] at (xaxis) {$x$};
\node[left] at (yaxis) {$y$};
\node[below left] {$O$};
\node at (0.8,0.3) {$P$};
\node at (1.3,-1.4) {$C_{+}$};
\node at (0.3,-1.4) {$C_{-}$};
\node at (-0.7,-1.8) {$C_{R}^{-}$};
\node at (2.3,-1.8) {$C_{R}^{+}$};
\node at (1.4,-0.5) {$C_{\rho}$};
\end{tikzpicture}    
    \caption{Integration domain for the integral in Equation \eqref{eq:It1_v2}. When integrating, we take limits $\rho \rightarrow 0$ for the radius of the inner arc, and $R \rightarrow \infty$ for the radius of the outer arc. The arc $C_{\rho}$ is centered around the divergence point $u^* = i \log(1-p(v,\tau))$, so that it lies outside the contour.}
    \label{fig:contour_integral2}
\end{figure}

Notice that the contour in Figure \ref{fig:contour_integral2} is the same as in Figure \ref{fig:contour_integral1}, but shifted by $\operatorname{Re}(u^*)$ on the real axis. This suggests the following change of variables: $u' = u - \operatorname{Re}(u^*)$. Equation \eqref{eq:It1_v2} then becomes
\begin{align}\label{eq:It1_v3}
    I_1(x) = &\, \exp\left(\frac{\lambda^* \tau}{2 \alpha}(\beta-\alpha-i \alpha\bar{\mu}_Y v - f(v) )\right)\frac{(2 f(v))^{\scaleobj{1.2}{\frac{\lambda^*}{\alpha}}}}{h(v,\tau)^{\scaleobj{1.1}{\frac{\lambda^*}{\alpha}+Q_s}}} \nonumber \\ & \boldsymbol{\cdot} \frac{1}{2\pi}\int_{-\infty}^{+\infty} \left(\frac{1}{1-|1-p(v,\tau)|e^{iu'}}\right)^{\scaleobj{1.1}{\frac{\lambda^*}{\alpha}+Q_s}}e^{-i (u' + \operatorname{Re}(u^*))x} du',
\end{align}
where $|\cdot|$ denotes the norm of a complex number. Therefore, since $0\leq |1-p(v,\tau)| \leq 1$ is a real number, we can relate the integral in Equation \eqref{eq:It1_v3} to the PMF of the negative binomial distribution. Finally, considering that 
\begin{equation}
    e^{-i \operatorname{Re}(u^*) x} = \left(\frac{1-p(v,\tau)}{|1-p(v,\tau)|}\right)^x,
\end{equation}
Equation \eqref{eq:It1_v3} immediately yields Equation \eqref{eq:It1}, completing the proof.
\end{proof}

\subsection{The Heston-Hawkes model}\label{sec:hawk}

When the jumps are given by the Hawkes process instead of the Q-Hawkes process, we recover the Heston-Hawkes (HH) model, that has been partially studied in \cite{ait-sahalia2015}, \cite{liu2019} and \cite{jing2021}, among others. Here we extend those studies by showing that the joint characteristic function of this model can be expressed as a system of ODEs, which needs to be solved numerically.

\subsubsection{The Hawkes process}

The Hawkes process is a multivariate point process first defined in the early seventies (\cite{Hawkes1971a},\cite{Hawkes1971b}), characterized by a stochastic intensity which is a linear function of past jumps. Mathematically, a one-dimensional\footnote{Since in this article we are only interested in the univariate Hawkes process, we describe only this particular case.} counting process $N_t$ is called a Hawkes process if its intensity is given by 
\begin{equation}\label{eq:hawkes_def}
    \lambda_t = \lambda^* + \int_{-\infty}^t \phi(t-s)\, dN_s,
\end{equation}
where $\lambda^*$ is the baseline intensity and $\phi(\cdot)$ is a positive and causal function of time belonging to the space of $L^1_{loc}$-integrable functions. Since $\phi(\cdot)$ determines the impact of each past jump through time, we refer to this function as the memory kernel from now on. The positivity of the memory kernel implies that the intensity increases with each jump, thus having a self-exciting or clustering effect. The more jumps occur, the more probable it is for them to appear again. 

In practice, the Hawkes process is defined so that this self-excitation does not cause its intensity to diverge, not even in an infinite amount of time. A common assumption that ensures the stationarity of the intensity process is that the memory kernel has an $L^1$-norm strictly smaller than 1. This common assumption is known in the literature as the {\it stability condition (H)}. 

As mentioned in the introduction, for most financial applications the memory kernel is assumed to be an exponential function of time. Under this setting, the Hawkes process can be shown to be a Markov process, and in particular an affine jump-diffusion process. The result is formalized in the following proposition.
\begin{prop}
    Consider the Hawkes process with exponential memory kernel $\phi(t) = \alpha e^{-\beta t}$, for $t,\alpha \in \mathbb{R}^+$ and $\beta \in \mathbb{R}^+\backslash\{0\}$. Taking\footnote{Notice that $t$ is only defined on the positive real line, while the integral in Equation \eqref{eq:hawkes_def} contains also the negative real line. Thus, following the literature, we appropriately define the negative part of the integral as in Equation \eqref{eq:hawk_neg_int}.}
    \begin{equation}\label{eq:hawk_neg_int}
        \int_{-\infty}^0 \alpha e^{\beta s}\, dN_s \coloneqq (\lambda_0-\lambda^*) e^{-\beta t},
    \end{equation}
    the couple $(N_t,\lambda_t)$ is a Markov process. Moreover, the intensity can be rewritten as
    \begin{equation}\label{eq:hawkes_diff}
        d\lambda_t = \beta(\lambda^*-\lambda_{t_{-}})\,dt + \alpha dN_t,
    \end{equation}
    with $t_{-} \coloneqq \lim_{\epsilon \rightarrow 0^+} t-\epsilon$ the moment right before $t$. Also, the stability condition in this model is satisfied iff $\beta>\alpha$.
\end{prop}
\begin{proof}
Given Equation \eqref{eq:hawk_neg_int}, Equation \eqref{eq:hawkes_diff} is obtained by direct differentiation of Equation \eqref{eq:hawkes_def}. The verification of the stability condition is obtained by the integration of the memory kernel.
\end{proof}

Since it is an AJD process, its characteristic function can be obtained by numerically solving a system of ODEs, derived from the results of \cite{duffie2000}. Such a system for the Hawkes process was obtained in \cite{errais2010}, to where we refer for those formulas.

\subsection{Characteristic function of the HH model}

Analogously to what we did for the Q-Hawkes process, we derive an expression for the characteristic function of $M_t$ in order to price options via the COS method. While a closed-form expression for this function is not available, in the next proposition we show that it can be obtained by solving numerically a system of ODEs.

\begin{prop}
     Let $\lambda_t$ be the intensity of a Hawkes process with clustering rate $\alpha \in \mathbb{R}^+$, expiration rate $\beta \in \mathbb{R}^+\backslash\{0\}$, and baseline intensity $\lambda^* \in \mathbb{R}^+$. Moreover, let $Y$ be a random variable with support in $(-\infty,\infty)$ and characteristic function $\psi_Y(\cdot)$. Then, the characteristic function of the pair $(\lambda_t,M_t)$ is given by
     \begin{equation}\label{eq:MtH_cf}
         \mathbb{E}[e^{i u \lambda_t + i v (M_t-M_s)}|\mathcal{F}_s] = e^{A(u,v,s) + B(u,v,s)\lambda_s },
     \end{equation}
     where the functions $A(u,v,\tau)$ and $B(u,v,\tau)$ satisfy the following system of ODEs
     \begin{align}
         \frac{\partial A(u,v,\tau)}{\partial \tau} &= \beta \lambda^* B(u,v,\tau), \\
         \frac{\partial B(u,v,\tau)}{\partial \tau} &= e^{\alpha B(u,v,\tau)}\psi_Y(v) -\beta B(u,v,\tau) - i v \bar{\mu}_Y - 1,
     \end{align}
     and the set of boundary conditions $A(u,v,t) = 0$, $B(u,v,t) = i u$.
\end{prop}
\begin{proof}
The proof goes along the same lines as the proof of Proposition 2.2 in \cite{errais2010}. First, notice that the infinitesimal generator $\mathcal{D}$ of the pair $(\lambda_t,M_t)$ is given by
\begin{equation}
    (\mathcal{D}g)(\lambda,M) = \beta(\lambda^*-\lambda)g_{\lambda}(\lambda,M)- \lambda\, \bar{\mu}_Y g_M(\lambda,M) + \lambda \int (g(\lambda+\alpha,M+x)-g(\lambda,M))d\nu(x),
\end{equation}
where $g_z$ denotes the partial derivative of $g$ with respect to $z$, and $\nu(\cdot)$ is the cumulative distribution function (CDF) of the random jump size $Y$.

Then, define $\psi_{\lambda,M}(u,v,t,s) \coloneqq \mathbb{E}[e^{i u \lambda_t + i v (M_t-M_s)} |\mathcal{F}_s]$, for $s\leq t$. Observe that, by definition, $\psi_{\lambda,M}(u,v,t,s)$ is a martingale with respect to the filtration generated by $\lambda_t$ and $M_t$. By the properties of the infinitesimal generator (see \cite{errais2010}), this implies
\begin{equation}\label{eq:hawk_infgen}
    \begin{dcases}
            \frac{\partial \psi_{\lambda,M}}{\partial s}(u,v,t,s) + (\mathcal{D}\psi_{\lambda,M})(u,v,t,s) = 0, \\ 
            \psi_{\lambda,M}(u,v,t,t) = e^{i u \lambda_t + i v M_t}
	\end{dcases}
\end{equation}
To finalize the proof, we point out the fact that, since the pair $(\lambda_t,M_t^H)$ is an AJD process, we can use the results from \cite{duffie2000} and assume
\begin{equation}\label{eq:hawk_duffie}
    \psi_{\lambda,M}(u,v,t,s) = e^{A(u,v,s) + B(u,v,s)\lambda_s}.
\end{equation}
Substituting Equation \eqref{eq:hawk_duffie} into Equation \eqref{eq:hawk_infgen} yields the desired result.
\end{proof}

\section{COS method for self-exciting jumps}\label{sec:cos}

We describe the methodology to price European and Bermudan options with self-exciting processes. As mentioned in the introduction, the method of choice is the COS method, first developed in \cite{fang2008}, due to the availability of the characteristic function of the Q-Hawkes process. For the pricing of European options using the COS method, we refer to \cite{fang2008} for all the details. For the case of Bermudan options, in this paper we adopt an approach based on \cite{fang2011}, with some adaptations to include the jump process. 

\begin{remark}
The COS method also allows for a fast computation of the Greeks, such as $\Delta$ and $\Gamma$ (see \cite{fang2008}). Due to the properties of the general jump-diffusion model presented in Section \ref{sec:general}, this also applies to the HQH and HH models. Other sensitivities with respect to the model parameters can be obtained by partial differentiation of either the characteristic function of the Heston model or Equation \eqref{eq:mt_char_ese}. 
\end{remark}

\subsection{COS method for pure-jump processes}\label{sec:cos_jump}

The COS method makes use of the cosine transform to expand the density function as a linear combination of cosine functions, whose coefficients depend on the characteristic function of the variable of interest. However, for a pure-jump process, the concept of density cannot be applied, due to the discontinuities of the probability distribution. Instead, the PMF must be used. However, since the PMF is not a continuous function by definition, its cosine transform is also not well-defined. Hence, in the following we employ the discrete cosine transform (DCT) in order to define a COS method for discrete random variables.

\begin{definition}
Let $X$ be a discrete random variable in $\mathbb{N}$, and define $p_X(n) \coloneqq \mathbb{P}[X = n]$, $n\in \mathbb{N}$. Then, the first $N$ values---i.e. $(p_X(0),\dots,p_X(N-1))$---can be expressed in terms of the cosine series expansion
\begin{equation}\label{eq:DCT}
     p_X(n) = \sideset{}{'}\sum_{k=0}^{N-1}  A_k \cos\left(k\pi\frac{2n+1}{2N}\right),
\end{equation}
where the $A_k$'s
\begin{equation}
    A_k = \frac{2}{N}\sum_{n=0}^{N-1}p_X(n)\cos\left(k\pi\frac{2n+1}{2N}\right),
\end{equation}
are the DCT of the $p_X(n)$'s and $\sideset{}{'}\sum$ means that the first term in the summation is multiplied by a half.
\end{definition}
In analogy with the COS method, we approximate $A_k$ by 
\begin{align}\label{eq:DCT_COS}
    A_k &\simeq \frac{2}{N}\operatorname{Re}\left(\sum_{n=0}^{\infty}p_X(n)e^{i\frac{k\pi}{N}(n+\frac{1}{2})}\right) \nonumber \\
    &= \frac{2}{N}\operatorname{Re}\left(\psi_X\left(\frac{k\pi}{N}\right)e^{i\frac{k\pi}{2N}} \right) \nonumber \\
    &\eqqcolon \hat{A}_k,
\end{align}
where $\psi_X(\cdot)$ is the characteristic function of the discrete random variable $X$. 

Therefore, by plugging Equation \eqref{eq:DCT_COS} into Equation \eqref{eq:DCT}, we can approximate the PMF of $X$ in terms of its characteristic function through the DCT.

\begin{remark}
The same procedure generalizes to multidimensional distributions. Just observe that, if the distribution encompasses both continuous and discrete random variables, then we apply the cosine transform to the continuous random variables and the DCT to the discrete ones.
\end{remark}

\subsubsection{Error analysis}

It is of interest to understand the scope of the numerical error that we are incurring with the discrete COS (DCOS) method. Notice that, due to the orthogonality properties of the cosine functions, Equation \eqref{eq:DCT} is exact. Therefore, the only source of error originates from approximating $A_k$ by $\hat{A}_k$. Clearly,
\begin{equation}
    \hat{A}_k-A_k = \frac{2}{N}\sum_{m=N}^{\infty}p_X(m)\cos\left(k\pi\frac{2m+1}{2N}\right).
\end{equation}
And, denoting by $\hat{p}_X(n)$ the COS approximation of $p_X(n)$,
\begin{equation}\label{eq:COS_err1}
    \hat{p}_X(n) - p_X(n) = \frac{2}{N} \sideset{}{'}\sum_{k=0}^{N-1}\sum_{m=N}^{\infty}p_X(m)\cos\left(k\pi\frac{2m+1}{2N}\right)\cos\left(k\pi\frac{2n+1}{2N}\right).
\end{equation}
Due to the orthogonality properties of the cosine functions, Equation \eqref{eq:COS_err1} can be greatly simplified. In particular, we can write, without loss of generality, $m = lN + j$, with $l \in \mathbb{N}\backslash\{0\}$ and $j \in \{0,1,\dots,N-1\}$. Introducing this in Equation \eqref{eq:COS_err1} and changing the order of summation gives
\begin{equation}
    \hat{p}_X(n) - p_X(n) = \frac{2}{N} \sum_{l=1}^{\infty}\sum_{j=0}^{N-1}p_X(lN+j)\sideset{}{'}\sum_{k=0}^{N-1}\cos\left(k\pi\frac{2(lN+j)+1}{2N}\right)\cos\left(k\pi\frac{2n+1}{2N}\right).
\end{equation}
Now it is easy to see that, if $l$ is even
\begin{equation}
    \sideset{}{'}\sum_{k=0}^{N-1}\cos\left(k\pi\frac{2(lN+j)+1}{2N}\right)\cos\left(k\pi\frac{2n+1}{2N}\right) =
    \left\{\begin{array}{ll}
         0, &  \text{if }j\neq n,\\
         \\
         N/2, & \text{if }j = n.
    \end{array}\right.
\end{equation}
And when $l$ is odd,
\begin{equation}
    \sideset{}{'}\sum_{k=0}^{N-1}\cos\left(k\pi\frac{2(lN+j)+1}{2N}\right)\cos\left(k\pi\frac{2n+1}{2N}\right) =
    \left\{\begin{array}{ll}
         0, &  \text{if }j\neq N-1-n,\\
         \\
         N/2, & \text{if }j = N-1-n.
    \end{array}\right.
\end{equation}
Combining these results with Equation \eqref{eq:COS_err1} leads to
\begin{equation}\label{eq:COS_err2}
    \hat{p}_X(n) - p_X(n) = \sum_{l=1}^{\infty}\left(p_X(2lN+n)+p_X(2lN-1-n)\right).
\end{equation}
From Equation \eqref{eq:COS_err2} we deduce two insights. First, since the PMF is nonnegative, $\hat{p}_X(n)\geq p_X(n)$. Thus, the DCOS method always overestimates the actual probabilities. Second, the order of convergence of the DCOS method depends on the decay rate of $p_X(\cdot)$ with $N$. When we are close to the tail of the distribution, the decay rate depends on the differentiability\footnote{Obviously, a PMF does not have a well-defined derivative of any order, since its argument is discrete. Thus, when we say differentiability of a PMF, we mean how many derivatives are nonzero assuming the argument is continuous.} of $p_X(\cdot)$. If the first $k$ derivatives are nonzero, then the order of convergence is $\mathcal{O}(N^{-k-1})$. Moreover, for infinitely differentiable functions, the order of convergence is exponential, i.e. $\mathcal{O}(e^{-\gamma k^r})$, for some $\gamma,r>0$.

We illustrate this with two examples: the discrete uniform distribution and the Poisson distribution. For the discrete uniform distribution in $\{0,1,\dots,M\}$, we have that
\begin{equation}
    p_X^{uniform}(n) = \left\{
    \begin{array}{ll}
        1/(M+1), &  \text{if }n\leq M,\\
        0, & \text{Otherwise.}
    \end{array}\right.
\end{equation}
Thus, the expected order of convergence for $n\leq M$ is $\mathcal{O}(N^{-1})$, since the first derivative is already zero everywhere. This can indeed be observed in Figure \ref{fig:COS_err_uni}, where we show the approximation error as a function of $N$ for $M = 500$ and $n = 24$.

On the other hand, the Poisson distribution is infinitely differentiable, thus we expect an exponential order of convergence for large enough $N$. This is confirmed in Figure \ref{fig:COS_err_poi}, where we plot the error in log-scale for the case $\lambda = 15 $ and $n = 12$. Indeed, the error decays exponentially until it reaches the order of magnitude of the machine error, from where it flattens. 

\begin{figure}
    \centering
    \begin{subfigure}{0.4\textwidth}
    \centering
    \includegraphics[width=\textwidth]{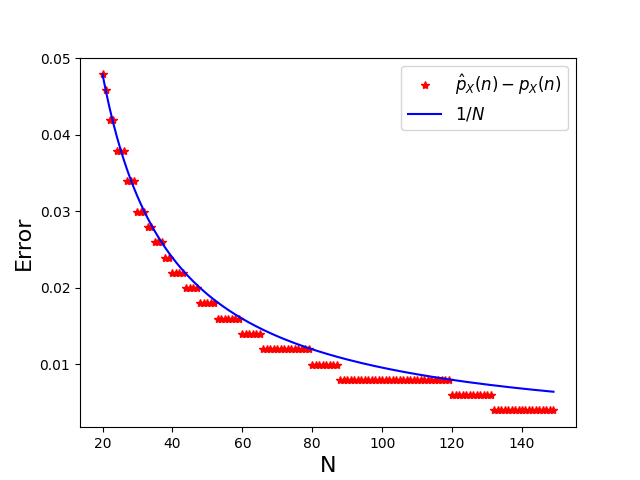}
    \caption{Error for discrete uniform distribution.}
    \label{fig:COS_err_uni}
    \end{subfigure}
    \begin{subfigure}{0.4\textwidth}
    \centering
    \includegraphics[width=\textwidth]{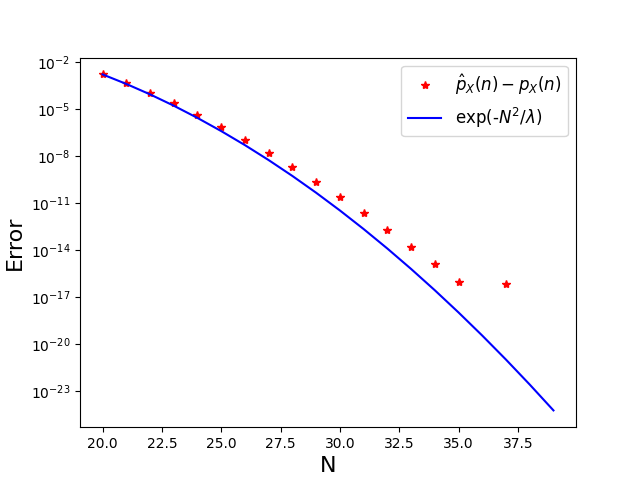}
    \caption{Log-error for Poisson distribution.}
    \label{fig:COS_err_poi}
    \end{subfigure}
    \caption{Numerical analysis of the order of convergence of the discrete COS method for the discrete uniform and Poisson distributions. Notice how the order of convergence for both cases is $\mathcal{O}(N^{-1})$ and $\mathcal{O}(e^{-N^2/\lambda})$, respectively, as expected from the differentiability of these distributions.}
\end{figure}

\subsection{Reducing dimensionality}\label{sec:cos_dim}

In \cite{fang2011}, the authors developed an adaption of the COS method to price Bermudan options under the Heston model. Given the nice analytical properties of the variance process--which follows a CIR process (e.g, \cite{oosterlee2019})--they could reduce the complexity of the COS method from $\mathcal{O}(N_1 N_2)$ to $\mathcal{O}(N_1)$, where $N_1$ and $N_2$ are the number of terms in the COS expansion of the log-asset price $X_t$ and the variance $V_t$, respectively. More precisely, since $f(X_t,V_t) = f(X_t|V_t)\cdot f(V_t)$, and the density of the variance is known, the problem is reduced to applying the COS expansion to the conditional density function of $X_t$.

However, this conditional density depends on past {\it and} current values of the variance. That is, $f(X_t,V_t|X_s,V_s) = f(X_t|X_s,V_s,V_t)\cdot f(V_t|V_s)$. Therefore, the conditional density depends on two values of the variance at two different times $s$ and $t$. This suggests that the dynamics of the ``bridge'' version of that process\footnote{That is, the dynamics of a stochastic process when both its initial and final states are fixed.} must be known first, and only then we can calculate the corresponding characteristic function. This is indeed what has been done for the Heston model. In \cite{broadie2006} they obtain the characteristic function of $X_t$, given $V_s$ and $V_t$, based on previous research from \cite{pitman1982} about Bessel bridges. Of course, the research about Bessel bridges is much broader and deeper than just a result about its characteristic function. Nevertheless, in the following we show that there is a simpler way to reduce the complexity of the COS method, given some analytical properties of the joint characteristic function of the processes involved.

For simplicity, assume that we want to estimate the joint density\footnote{Here we assume that the density exists, i.e. we are dealing with continuous random variables. However, the results from this section also apply for the estimation of the PMF of a discrete random variable. In that case, we simply use the discrete COS method described in Section \ref{sec:cos_jump}.} of two stochastic processes $X_t$ and $Y_t$, which we denote by $f_{X,Y}(X_t,Y_t)$. Let's further assume that the characteristic function $\psi_{X,Y}(u,v) = \mathbb{E}[e^{i u X_t + i v Y_t}]$ is known. Then, $f_{X,Y}(X_t,Y_t)$ can be approximated via the COS method as
\begin{equation}
    f_{X,Y}(X_t,Y_t) \approx \sideset{}{'}\sum_{k_x=0}^{N_x-1}\sideset{}{'}\sum_{k_y=0}^{N_y-1} A_{k_x,k_y}\cos\left(k_x \pi \frac{X_t-a_x}{b_x-a_x} \right)\cos\left(k_y \pi \frac{Y_t-a_y}{b_y-a_y} \right),
\end{equation}
with 
\begin{equation}
    A_{k_x,k_y} = \frac{2}{b_x-a_x}\frac{2}{b_y-a_y}\iint_{\mathbb{R}^2} f_{X,Y}(x,y)\cos\left(k_x \pi \frac{x-a_x}{b_x-a_x} \right)\cos\left(k_y \pi \frac{y-a_y}{b_y-a_y} \right) dx dy.
\end{equation}
The coefficients $A_{k_x,k_y}$ can be written in terms of the characteristic function $\psi_{X,Y}$ (see \cite{ruijter2012}), but what is relevant for us is that this algorithm has complexity $\mathcal{O}(N_x N_y)$. Next, we demonstrate how this complexity can be reduced to $\mathcal{O}(N_x)$ under certain conditions. 

\begin{prop}\label{prop:cos_reduction}
     Let $X_t$ and $Y_t$ be two stochastic processes defined on $\mathbb{R}^2$, with joint density function denoted by $f_{X,Y}(X_t,Y_t)$. Assume that the characteristic function $\psi_{X,Y}(u,v) = \mathbb{E}[e^{i u X_t + i v Y_t}]$ is known. Moreover, the integral $\int_{\mathbb{R}}\psi_{X,Y}(u,v)e^{-i v y} dv$ is also known. Then, $f_{X,Y}(X_t,Y_t)$ can be approximated via the COS method as
     \begin{equation}\label{eq:cos_reduced}
    f_{X,Y}(X_t,Y_t) \approx \sum_{k_x=0}^{N_x-1} \tilde{A}_{k_x}(Y_t)\cos\left(k_x \pi \frac{X_t-a_x}{b_x-a_x}\right),
    \end{equation}
    with 
    \begin{equation}\label{eq:cos_red_coeff}
    \tilde{A}_{k_x}(Y_t) = \frac{2}{b_x-a_x}\frac{1}{2\pi}\operatorname{Re}\left(e^{-i\frac{k_x\pi a_x}{b_x-a_x}}\int_{\mathbb{R}}\psi_{X,Y}\left(\frac{k_x \pi}{b_x-a_x},v\right) e^{- i v Y_t} dv\right).
    \end{equation}
\end{prop}
\begin{proof}
We approximate $f_{X,Y}$ by expanding the cosine series only in $X$. That is, we apply Equation \eqref{eq:cos_reduced}, where
\begin{equation}
    \tilde{A}(Y_t) = \frac{2}{b_x-a_x}\int_{\mathbb{R}}f(x,Y_t)\cos\left(k_x \pi \frac{x-a_x}{b_x-a_x} \right) dx.
\end{equation}
Now the goal is to write $\tilde{A}_{k_x}(Y_t)$ in terms of $\psi_{X,Y}(u,v)$. First, recalling the definition of the inverse Fourier transform 
\begin{equation}
    \tilde{A}_{k_x}(Y_t) = \frac{2}{b_x-a_x}\frac{1}{(2\pi)^2}\int_{\mathbb{R}}\cos\left(k_x \pi \frac{x-a_x}{b_x-a_x} \right)\left(\iint_{\mathbb{R}^2} \psi_{X,Y}(u,v) e^{-i u x - i v Y_t} du dv\right)dx.
\end{equation}
Next, we change the order of integration, substitute $\cos(x) = \operatorname{Re}(e^{ix})$ and apply the definition of the Dirac delta function from Equation \eqref{eq:dirac_def}. Thus,
\begin{equation}
    \tilde{A}_{k_x}(Y_t) = \frac{2}{b_x-a_x}\frac{1}{2\pi}\operatorname{Re}\left(e^{-i\frac{k_x\pi a_x}{b_x-a_x}}\iint_{\mathbb{R}^2}\psi_{X,Y}(u,v) e^{- i v Y_t} \delta\left(u-\frac{k_x \pi}{b_x-a_x}\right)du dv\right).
\end{equation}
Finally, due to the properties of the Dirac delta function, we obtain Equation \eqref{eq:cos_red_coeff}, completing the proof.
\end{proof}
\begin{remark}
When $k_x=0$ in Equation \eqref{eq:cos_red_coeff}, the integral in $\tilde{A}_0(Y_t)$ is simply the marginal density of $Y_t$. Thus, a requirement to reduce the dimensionality of the COS method is that some of the marginal densities can be computed analytically. Such a requirement may sound trivial, as in case both integrals are available---with respect to $u$ and $v$---the joint density $f_{X,Y}$ has a closed-form solution and there is no need to use the cosine expansion. Proposition \ref{prop:cos_reduction} provides a simple way to exploit such analytical properties when only a subset of variables admits an analytical density function.
\end{remark}

\subsection{COS method for Bermudan options with jumps}\label{sec:cos_ber}
In this section, we use the COS method to derive the pricing formula for Bermudan options under the Heston jump-diffusion models defined in Section \ref{sec:jump_diff}. We mainly extend the methodology of \cite{fang2011} to cope with an extra dimension, given in this case by the jump process. 

In \cite{fang2011}, the analytical properties of the Heston model were used to reduce the dimensionality of the cosine expansion. Therefore, only the cosine expansion of the log-asset price was needed. In this article, we also deal with the jump process, which could add an extra dimension. This is particularly the case for the Hawkes process, since its characteristic function can only be computed numerically. Regarding the Q-Hawkes process, we have seen in Section \ref{sec:jump_diff} that the characteristic function of the compensated jump term can be integrated. Thus, we can use the results from Section \ref{sec:cos_dim} to further reduce the dimensions of the COS method. In particular, for the HQH model, the COS method is applied only to the log-asset price, in a similar way as in \cite{fang2011}. On the other hand, for the HH model, we must apply the cosine expansion for the log-asset price and the jump intensity $\lambda_t$. In the remainder of this section, we briefly describe how to price Bermudan options via the COS method, emphasizing the differences between the HQH and HH processes.

Assume the option allows for $M$ exercise opportunities at times $\mathcal{T} \coloneqq \{t_m, t_m<t_{m+1}|m = 0,1, \dots, M\}$, with $t_0$ corresponding to the beginning of the contract and $t_M = T$ corresponding to maturity. Furthermore, assume for simplicity that the exercise times are equally spaced between $t_0$ and $t_M$. That is, $\Delta t \coloneqq t_{m+1}-t_{m}$ is constant. Denoting by $\bm{Y}_m \coloneqq (X_{t_m},V_{t_m},\lambda_{t_m})$ the triplet containing the values of the log-asset price, variance and jump intensity at time $t_m$, respectively, the value of the Bermudan option at $t_m$ can be written as
\begin{equation}
    v(\bm{Y}_m,t_m) = 
    \left\{ \begin{array}{ll}
 g(X_{t_m},t_m) & m = M,\\
 \\
\max(c(\bm{Y}_m,t_m),g(X_{t_m},t_m)) & m = 1,\dots,M-1, \\
\\
 c(\bm{Y}_m,t_m) & m = 0,\end{array} \right.
\end{equation}
where $g(X_t,t)$ is the payoff of the option at time $t$ and $c(\bm{Y},t)$ is the continuation value.

Due to the Markov property, the continuation value is given by
\begin{equation}
    c(\bm{Y}_m,t_m) = e^{-r\Delta t}\mathbb{E}_{t_m}[c(\bm{Y}_{m+1},t_{m+1})],
\end{equation}
where $\mathbb{E}_t[\cdot] \coloneqq \mathbb{E}[\cdot|\mathcal{F}_t]$ is the expectation conditioned on the filtration at time $t$. Writing the expectation in integral form, the continuation value becomes
\begin{equation}\label{eq:cont_val}
    c(\bm{Y}_m,t_m) = e^{-r \Delta t}\iiint_{\mathbb{R}^3} v(x,V,\lambda,t_{m+1}) f(x,V,\lambda|\bm{Y}_m) dx dV d\lambda,
\end{equation}
with $f(\cdot|\bm{Y}_m)$ the joint density\footnote{Notice that, for the Q-Hawkes process, we use the PMF instead of the density function for the intensity $\lambda_t$.} function conditioned on the values of the processes at time $t_m$. For the HQH model, the density\footnote{For the HQH process, we prefer to write the density as a function of the activation number $Q_t$, since it makes derivations easier. However, the change is trivial due to Equation \eqref{eq:esep_inten}. That also means the integral over $\lambda$ is replaced by a summation over $Q$.} can be approximated with the COS method as
\begin{equation}\label{eq:dens_QHJD}
    f(x,V,Q|\bm{Y}_{m}) \simeq \sideset{}{'}\sum_{k=0}^{N_x-1} A_k(V,Q|\bm{Y}_{m}) \cos\left(k\pi \frac{x-a_x}{b_x-a_x} \right),
\end{equation}
where
\begin{equation}\label{eq:Ak_qhjd}
    A_k(V,Q|\bm{Y}_m) = \frac{2}{b_x-a_x}\operatorname{Re}\left(e^{i k \pi \frac{x_m-a_x}{b_x-a_x}}\Psi_{Q}\left(\frac{k \pi}{b_x-a_x},Q,Q_m\right)\Psi_{V}\left(\frac{k \pi}{b_x-a_x},V,V_m\right)\right),
\end{equation}
\begin{equation}\label{eq:int_cf_Q}
    \Psi_{Q}(u,Q,Q_m) = \frac{1}{2\pi}\int_{\mathbb{R}}\psi_{Q,M}(v,u|Q_m)e^{-i v Q} dv,
\end{equation}
and 
\begin{equation}\label{eq:int_cf_V}
    \Psi_{V}(u,V,V_m) = \frac{1}{2\pi}\int_{\mathbb{R}}\psi_{Heston}(u,w|X_m=0,V_m)e^{-i w V}dw.
\end{equation}
Notice that Equations \eqref{eq:int_cf_Q} and \eqref{eq:int_cf_V} are known analytically. Equation \eqref{eq:int_cf_Q} was derived in Proposition \ref{prop:int_cf_mt}, while Equation \eqref{eq:int_cf_V} is the same as Equation (30) in \cite{fang2011}.

Plugging Equation \eqref{eq:dens_QHJD} into Equation \eqref{eq:cont_val} yields
\begin{align}\label{eq:cont_val2}
    c(\bm{Y}_m,t_m) \simeq e^{-r \Delta t}\operatorname{Re}&\left(
    \sideset{}{'}\sum_{k=0}^{N_x-1}e^{i k \pi \frac{x_m-a_x}{b_x-a_x}}\sum_{Q=0}^{\infty}\Psi_{Q}\left(\frac{k \pi}{b_x-a_x},Q,Q_m\right)\int_0^{\infty}\Psi_{V}\left(\frac{k \pi}{b_x-a_x},V,V_m\right)\right. \nonumber \\
    &\phantom{(\sideset{}{'}\sum_{k=0}^{\infty}}\boldsymbol{\cdot}\left.\left(\int_{-\infty}^{+\infty}v(x,V,Q,t_{m+1})\cos\left(k\pi \frac{x-a_x}{b_x-a_x} \right)dx\right) dV\right).
\end{align}
The innermost integral in Equation \eqref{eq:cont_val2}---i.e. the integral of the log-asset price---can be further split into two integrals following the procedure of \cite{fang2011}, which can be solved analytically for the case of European put and call options. Finally, truncating the summation over $Q$ and applying a quadrature rule to the integral over $V$ gives 
\begin{align}\label{eq:cont_val3}
    c(\bm{Y}_m,t_m) \simeq e^{-r \Delta t}\operatorname{Re}&\left(
    \sideset{}{'}\sum_{k=0}^{N_x-1}e^{i k \pi \frac{x_m-a_x}{b_x-a_x}}\sum_{Q=0}^{N_Q-1}\Psi_{Q}\left(\frac{k \pi}{b_x-a_x},Q,Q_m\right)\sum_{j=1}^{N_V}\omega_j\Psi_{V}\left(\frac{k \pi}{b_x-a_x},v_j,V_m\right)\right. \nonumber \\
    &\phantom{(\sideset{}{'}\sum_{k=0}^{\infty}}\boldsymbol{\cdot}\left.\left(\int_{-\infty}^{+\infty}v(x,v_j,Q,t_{m+1})\cos\left(k\pi \frac{x-a_x}{b_x-a_x} \right)dx\right) dV\right),
\end{align}
where $\{v_j\}_{j=1}^{N_V}$ are the quadrature nodes and $\{\omega_j\}_{j=1}^{N_V}$ the quadrature weights. 

In order to compute the prices of Bermudan options, we evaluate $c(\bm{Y}_m,t_m)$ via Equation  \eqref{eq:cont_val3} recursively for $m=M,M-1,\dots,0$. For Bermudan put and call options, we can take advantage of FFT methods to compute the integral of the log-asset price, resulting in a computational complexity of $\mathcal{O}(M N_x\log(N_x) N_Q^2 N_V^2)$ for the recursive loop. For more complex payoffs whose integral cannot be solved analytically, the computational complexity is given by $\mathcal{O}(M N_x^2 N_Q^2 N_V^2)$.

Following a similar procedure for the HH model, we can also obtain a numerical approximation for the continuation value. The differences are that the intensity is continuous---so the integral needs to be numerically approximated---and the characteristic function is not known analytically, so the COS method is applied to both the log-asset price and the jump intensity. Applying these changes gives us the formula for the continuation value under the HH process as follows:
\begin{align}\label{eq:cont_val_H}
    c(\bm{Y}_m,t_m) \simeq e^{-r \Delta t}\operatorname{Re}&\left(
    \sideset{}{'}\sum_{k=0}^{N_x-1}e^{i k \pi \frac{x_m-a_x}{b_x-a_x}}\sum_{j_H=1}^{N_H}\omega_{j_H}\Psi^H_{\lambda}\left(\frac{k \pi}{b_x-a_x},\lambda_{j_H},\lambda_m\right)\sum_{j_V=0}^{N_V}\omega_{j_V}\Psi_{V}\left(\frac{k \pi}{b_x-a_x},v_{j_V},V_m\right)\right. \nonumber \\
    &\phantom{(\sideset{}{'}\sum_{k=0}^{\infty}}\boldsymbol{\cdot}\left.\left(\int_{-\infty}^{+\infty}v(x,v_{j_V},\lambda_{j_H},t_{m+1})\cos\left(k\pi \frac{x-a_x}{b_x-a_x} \right)dx\right) dV\right),
\end{align}
with
\begin{align}\label{eq:int_cf_H}
    \displaystyle
    \Psi_{\lambda}^H(u,\lambda,\lambda_m) = \frac{1}{b_H-a_H} \sideset{}{'}\sum_{k_H=0}^{N_{\lambda}-1}& \cos\left(k_H\pi \frac{\lambda-a_H}{b_H-a_H} \right)\nonumber \\
    & \boldsymbol{\cdot} \left(\psi_{M,\lambda}^H\left(u,\frac{k_H\pi}{b_H-a_H}|\lambda_m\right)e^{-i\frac{k_H\pi a_H}{b_H-a_H}}+\psi_{M,\lambda}^H\left(u,-\frac{k_H\pi}{b_H-a_H}|\lambda_m\right)e^{i\frac{k_H\pi a_H}{b_H-a_H}}\right),
\end{align}
and $\{\lambda_{j_H},\omega_{j_H}\}_{j_H=1}^{N_H}$ the quadrature nodes and weights for the intensity, respectively.

\begin{remark}
Comparing Equations \eqref{eq:cont_val3} and \eqref{eq:cont_val_H}, it is clear that the main difference is the extra COS expansion in the jump intensity for the Hawkes process. The system of ODEs that yields the characteristic function in Equation \eqref{eq:int_cf_H} needs to be solved for every combination of $k_x$, $\pm k_H$, $\lambda_{j_H}$ and $\lambda_m$, giving a total of  $2 \times N_x \times N_{\lambda} \times N_H^2$ evaluations. In comparison, Equation \eqref{eq:int_cf_Q} does not require solving any system of ODEs numerically and needs to be evaluated $N_x \times N_Q^2$ times.
\end{remark}
\section{Numerical results}\label{sec:results}

In this section, we perform pricing experiments for European and Bermudan options using the methodology derived in the previous sections. For all options we compare the prices of the HQH and HH models. Moreover, we also compare the results with the Bates model \cite{bates1996}, since this corresponds to the case in which the jumps are independent and there is no self-excitation. In this way, it is possible to benchmark the impact of self-excitation on the option prices.

For all derivatives, we consider two scenarios, \textit{Scenario A} and \textit{Scenario B}, whose parameters' settings are given in Table \ref{tab:parameters}. Notice that most of the parameters are the same in both cases, except for those related to the jump process. This is done in order to study the impact of the different forms of self-exciting behaviour on the option prices. Furthermore, we assume that the log-jump sizes are normally distributed with mean $\mu_Y$ and standard deviation $\sigma_Y$. 

\begin{table}[ht]
  \caption{Parameter setting for each scenario. Notice that only some of the parameters related to the jump process change between scenarios, leaving the diffusion component unchanged.}
  \centering
  \begin{threeparttable}
    \begin{tabular}{cccccccccccccc}\midrule\midrule 
      & $\alpha$ & $\beta$ & $\lambda^*$ & $Q_0$ & $\mu_Y$ & $\sigma_Y$ & $S_0$ & $r$ & $V_0$ & $\kappa$ & $\theta$ & $\eta$ & $\rho$ \tnote{*}\\ \cmidrule(l r){1-14}
        \textit{Scenario A} & $2$ & $3$ & $1.1$ & $2$ & $-0.3$ & $0.4$ & $9$ & $0.1$ & $0.0625$ & $5$ & $0.16$ & $0.9$ & $0.1$\\ \cmidrule(l r){1-14}
     \textit{Scenario B} & $2.9$ & $3$ & $1.1$ & $2$ & $0.3$ & $0.4$ & $9$ & $0.1$ & $0.0625$ & $5$ & $0.16$ & $0.9$ & $0.1$\\
      \midrule\midrule
    \end{tabular}

  \begin{tablenotes}
  \item[*] Refer to Equation \eqref{eq:esepjd_sde} for the meaning of each parameter.
  \end{tablenotes}
  \end{threeparttable}
  \label{tab:parameters}
\end{table}

For the sake of comparison, we assume the same setting for the HQH and HH processes. This is possible because both models have the same number of parameters and even the same interpretations. Thus, for the Hawkes process we set $\lambda_0 = \lambda^* + \alpha Q_0$. On the other hand, comparing with the Bates model is not so straightforward, because the set of parameters does not match. Since the Bates model has a constant jump intensity, we choose it so that the expectation of the number of jumps in one year is the same as for the other models.

Before moving to the pricing experiments, in Figure \ref{fig:density} we present the density functions for the HQH and HH models. This also helps us understanding the results in derivative pricing, given that they clearly rely on the distribution of the underlying. We observe that, in Scenario \textit{A}, the distributions are considerably similar. Therefore, it is expected that the European option prices will also be similar. On the other hand, in Scenario \textit{B} we observe larger differences, with the HQH model putting larger weights to the left side of the distribution. Further experiments, not shown here for reasons of space, indicate that the main contributor to the differences between Figures \ref{fig:density_A} and \ref{fig:density_B} is the mean jump size $\mu_Y$, with the clustering rate $\alpha$ having a smaller impact than that of $\mu_Y$ on these discrepancies. 

\begin{figure}[ht]
    \centering
    \begin{subfigure}{0.4\textwidth}
    \centering
    \includegraphics[width=\textwidth]{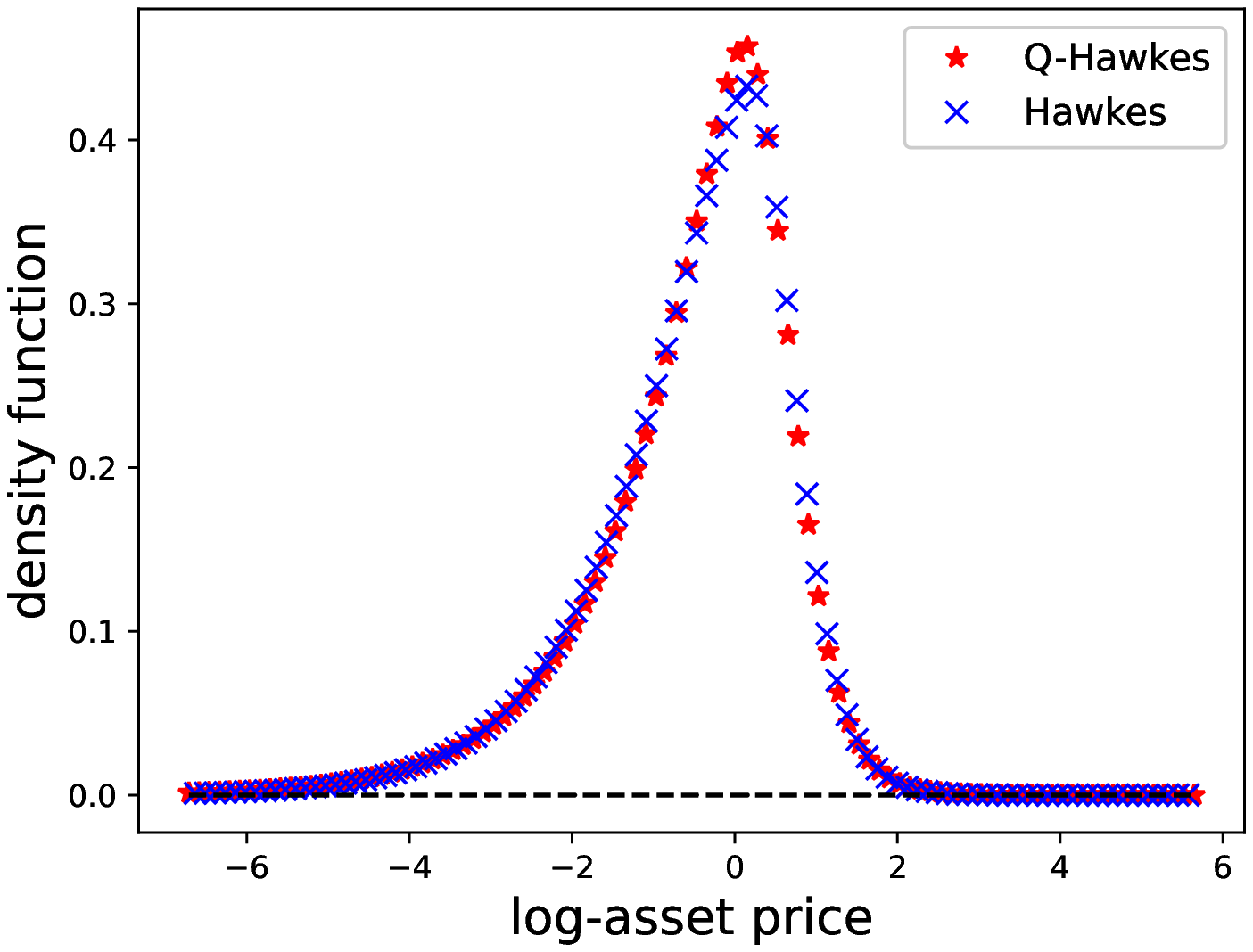}
    \caption{Scenario \textit{A}.}
    \label{fig:density_A}
    \end{subfigure}
    \begin{subfigure}{0.4\textwidth}
    \centering
    \includegraphics[width=\textwidth]{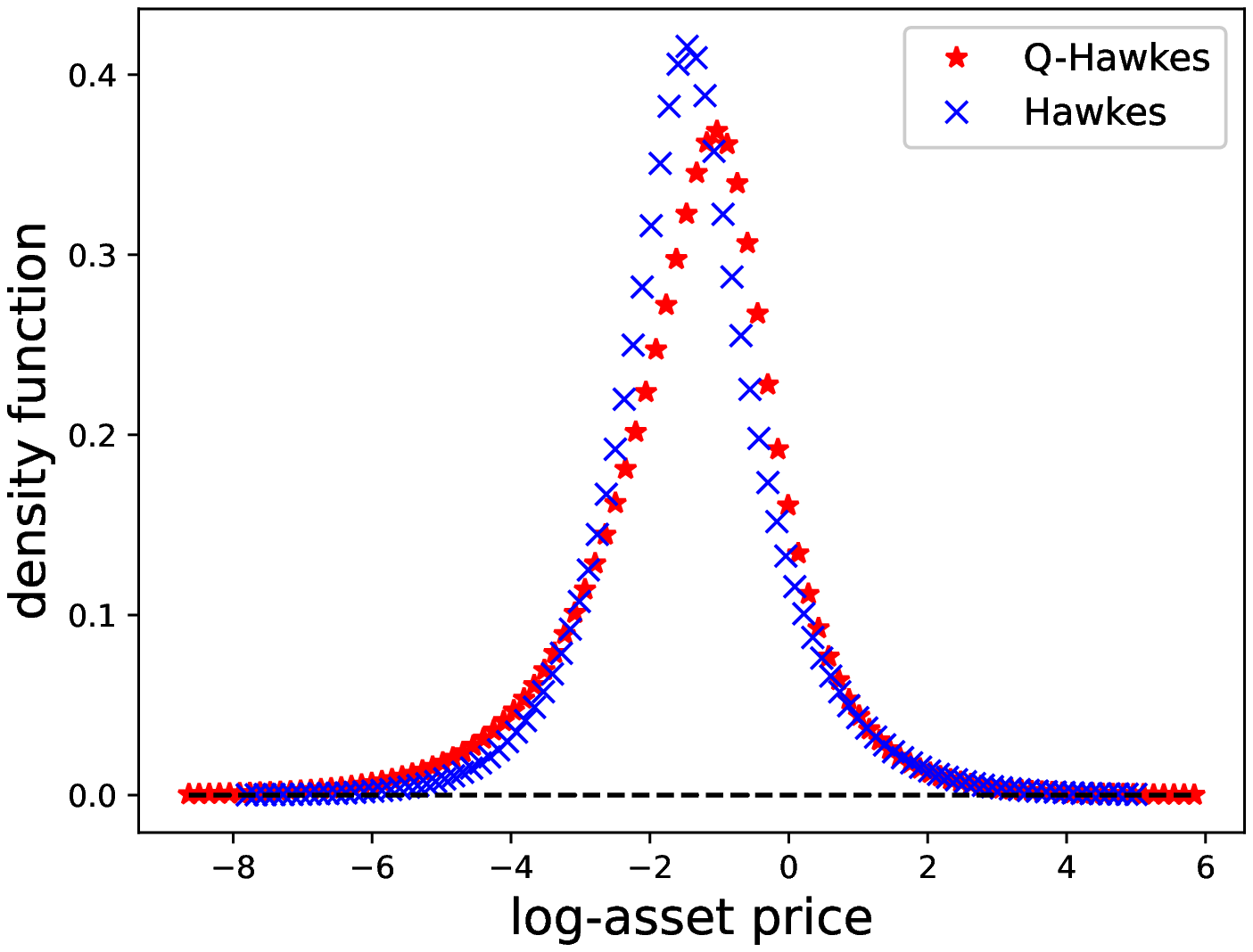}
    \caption{Scenario \textit{B}.}
    \label{fig:density_B}
    \end{subfigure}
    \caption{Comparison of the log-asset price density function at time $T = 1$, recovered via the COS method for the HQH and HH models in Scenarios \textit{A} and \textit{B}.}
    \label{fig:density}
\end{figure}

\begin{remark}
The code was implemented in Python (version 3.8.10), optimized through the library for scientific computing NumPy. Experiments were run in an Intel(R) Core(TM) i7-7700HQ CPU @ 2.80GHz processor.
\end{remark}

\subsection{Pricing European options}

We start with the pricing of European put options. As stated before, we compare the prices for the HQH, HH and Bates models under Scenarios \textit{A} and \textit{B}. The comparison is done in terms of the Black-Scholes implied volatility, since this quantity is sensitive to variations in the option price, specially for out-of-the-money options. 

In Figures \ref{fig:strike} and \ref{fig:maturity}, we report the implied volatility prices for a chosen set of expiry dates $T$ and strike prices $K$ under all models for Scenarios \textit{A} and \textit{B}, respectively. Clearly, the results are very different under both scenarios. In scenario \textit{A}, the Bates model always gives the highest implied volatility, although the differences with respect to the other models change significantly with respect to strike and maturity. On the other hand, in scenario \textit{B} we observe some regions where the Bates model gives the highest volatilities, but also others where the opposite happens. It is interesting to notice that the HH model always yields higher implied volatilities than the HQH model, regardless of the scenario we consider. These differences are, of course, much smaller in scenario \textit{A}, as expected by looking at the densities in Figure \ref{fig:density}. Notwithstanding these quantitative discrepancies, the qualitative behaviour is almost the same in both models, producing a broader range of volatility smiles than the Bates model. 

\begin{figure}[ht]
    \centering
    \begin{subfigure}{0.4\textwidth}
    \centering
    \includegraphics[width=\textwidth]{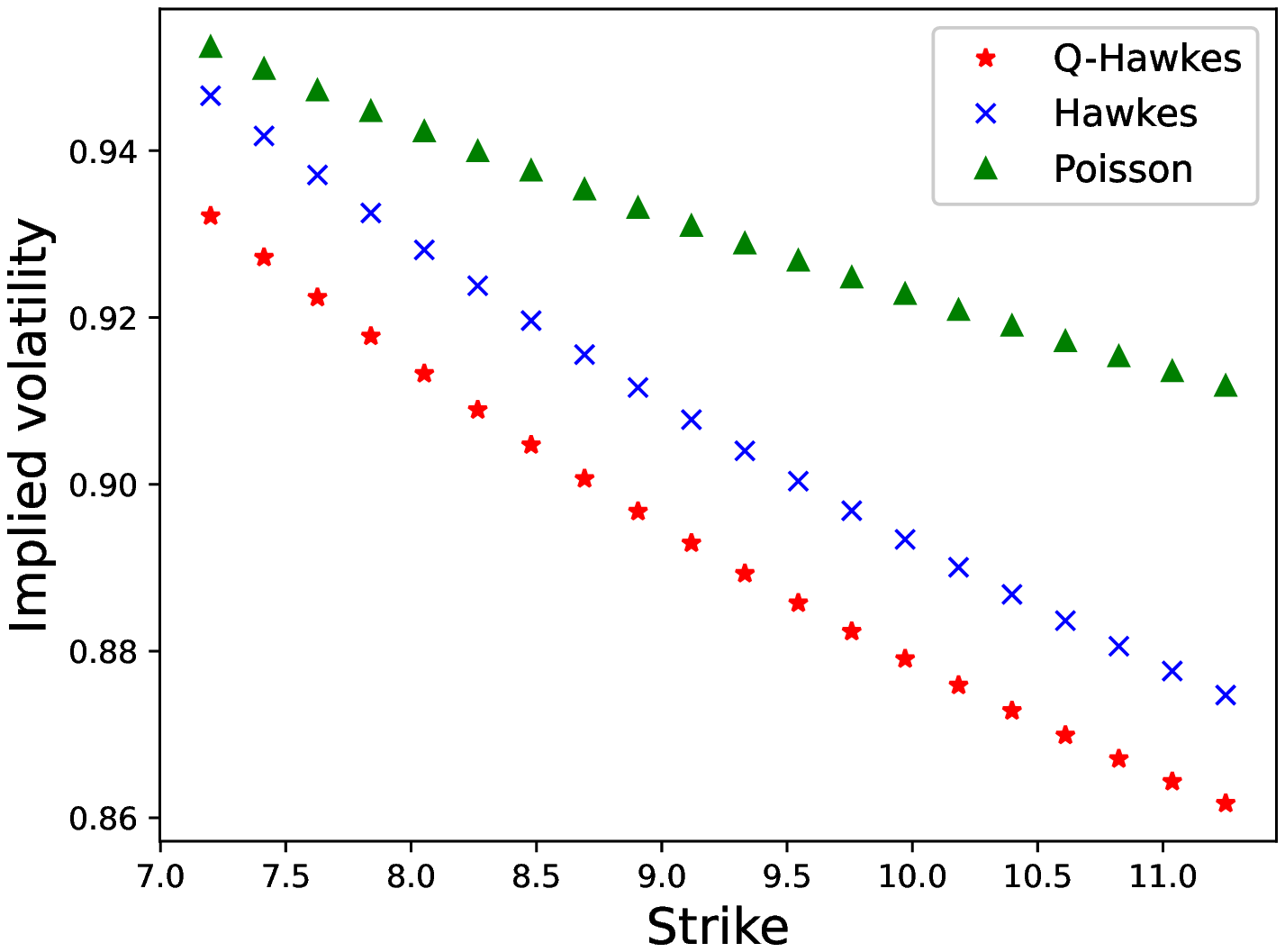}
    \caption{Scenario \textit{A}.}
    \label{fig:strike_A}
    \end{subfigure}
    \begin{subfigure}{0.4\textwidth}
    \centering
    \includegraphics[width=\textwidth]{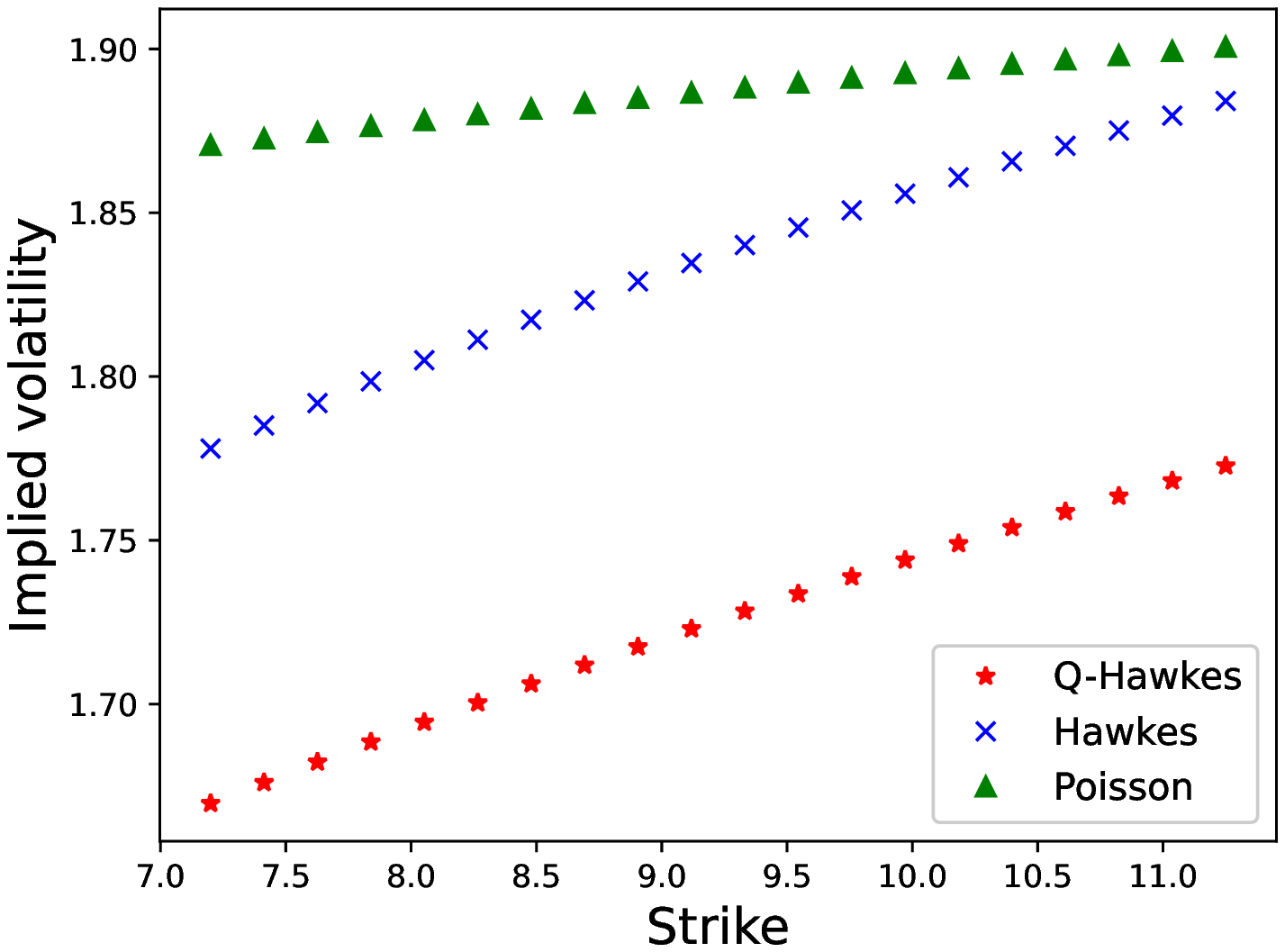}
    \caption{Scenario \textit{B}.}
    \label{fig:strike_B}
    \end{subfigure}
    \caption{Implied volatility as a function of the strike obtained from European put options for the HQH, HH and Bates models under Scenarios \textit{A} and \textit{B}. Maturity is set to one year.}
    \label{fig:strike}
\end{figure}

\begin{figure}[ht]
    \centering
    \begin{subfigure}{0.4\textwidth}
    \centering
    \includegraphics[width=\textwidth]{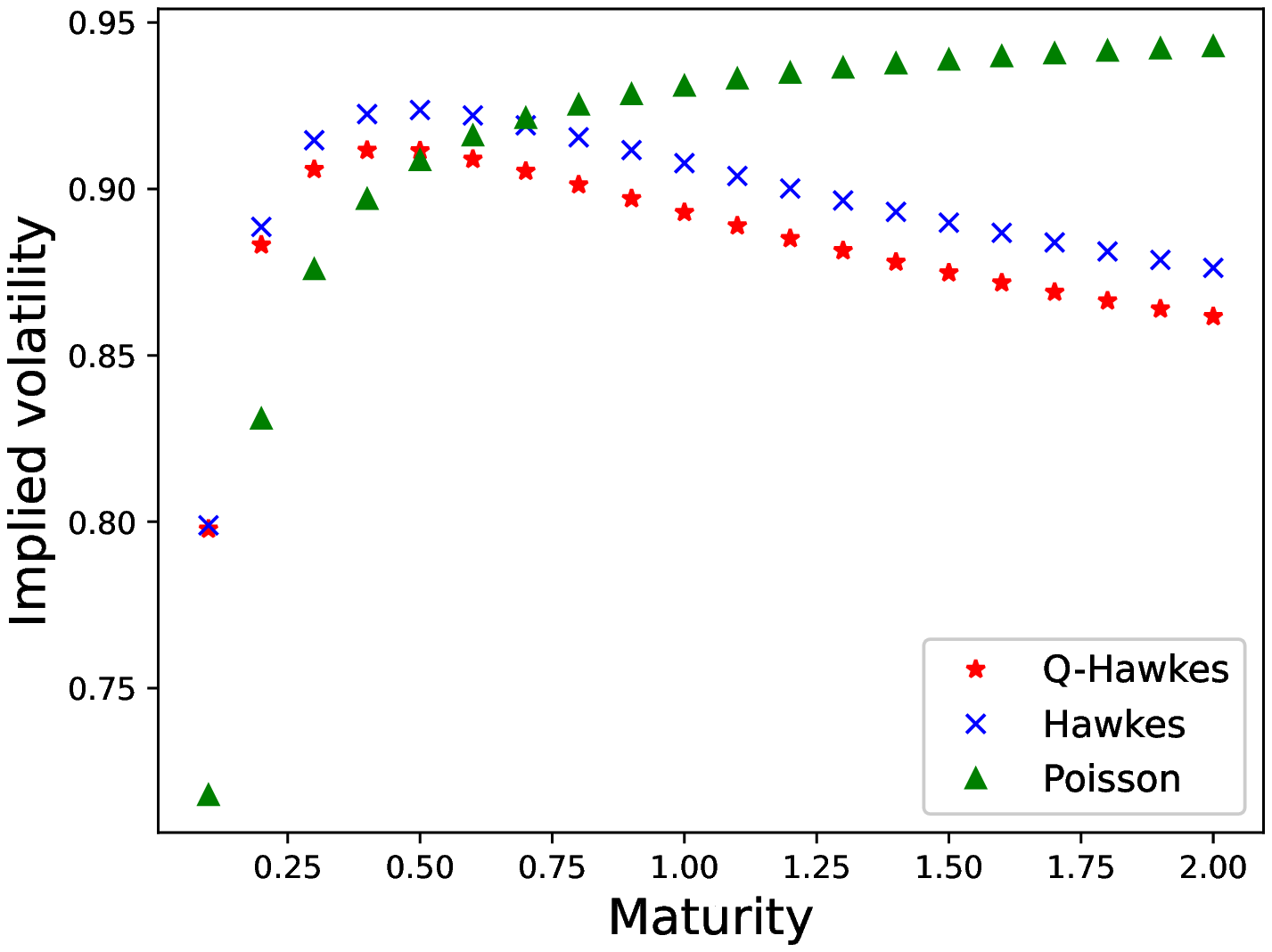}
    \caption{Scenario \textit{A}.}
    \label{fig:maturity_A}
    \end{subfigure}
    \begin{subfigure}{0.4\textwidth}
    \centering
    \includegraphics[width=\textwidth]{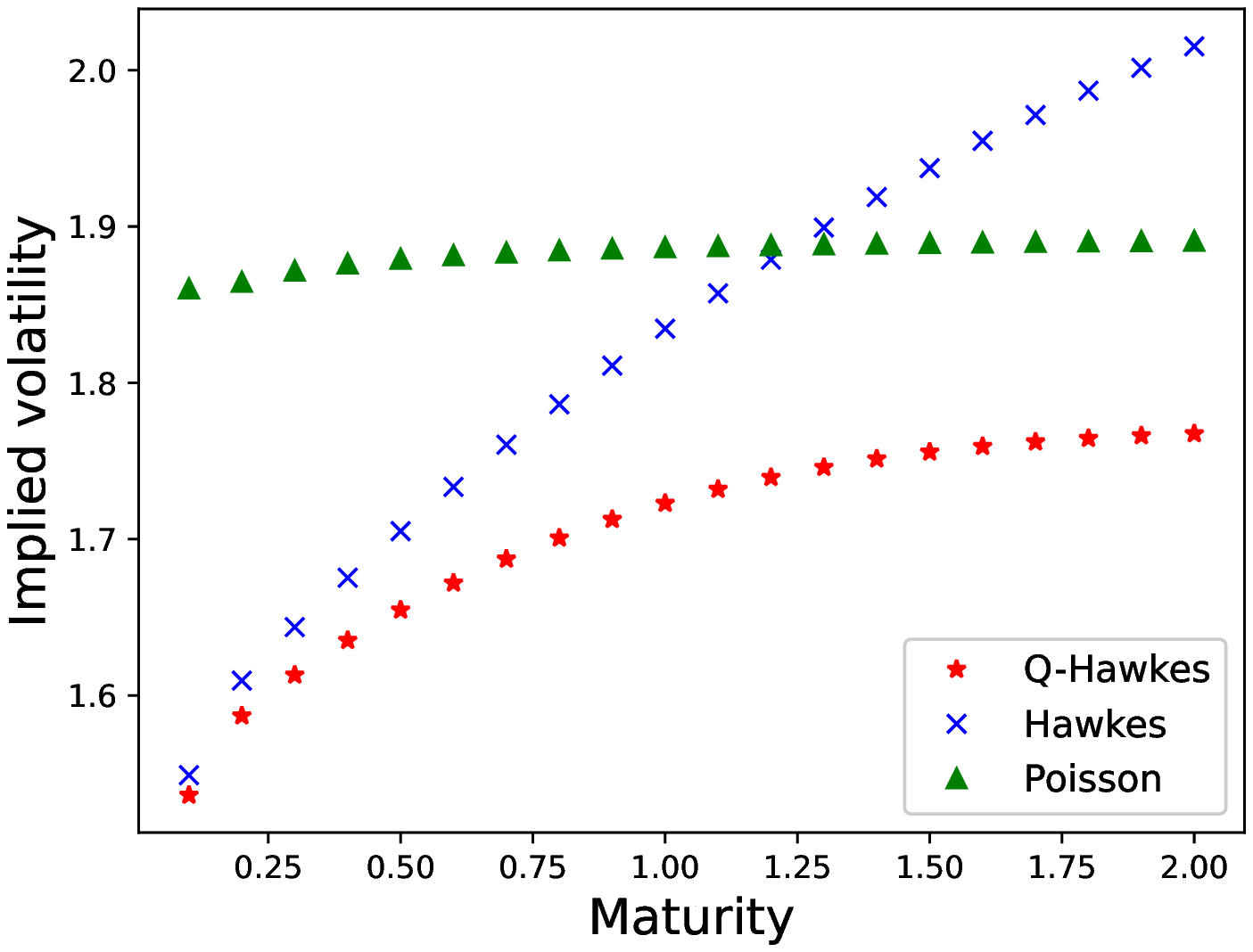}
    \caption{Scenario \textit{B}.}
    \label{fig:maturity_B}
    \end{subfigure}
    \caption{Implied volatility as a function of the maturity obtained from European put options for the HQH, HH and Bates models under Scenarios \textit{A} and \textit{B}. The strike price is chosen so that the option is at-the-money.}
    \label{fig:maturity}
\end{figure}

Apart from the implied volatilities, we also compare the computational time spent by each model pricing these options with the COS method. So as to have a robust estimate of the computational times, we price European put options on a grid of 21 strikes and 20 maturities---with a total of 420 evaluations---and calculate the average time over 50 runs. The outcome is in Table \ref{tab:put_speedup}, where apart from the total time we also show the speedup with respect to the HH model. We see that the HQH and Bates models outperform the HH model by a factor of 12.36 and 16.43, respectively. The results are analogous for Scenarios \textit{A} and \textit{B}, so we only present the results for Scenario \textit{B}.

\begin{table}[ht]
  \caption{Average computational time (in seconds) over 50 runs required for each model to calculate the grid of European put prices.}
  \centering
    \begin{tabular}{lccccc}\midrule\midrule 
      & Hawkes & Q-Hawkes & Bates\\ \cmidrule(l r){1-4}
        Time & 1.457 & 0.118 & 0.089\\ \cmidrule(l r){1-4}
     Speedup & 1 & 12.36 & 16.43\\ 
      \midrule\midrule
    \end{tabular}
  \label{tab:put_speedup}
\end{table}

\subsection{Pricing Bermudan options}

We now consider the pricing of Bermudan put options. While plain vanilla options only depend on the characteristics of the underlying at the expiry date of the contract, Bermudan options give the holder the possibility of early exercise. Therefore, the shape of the underlying trajectories also plays a role in the final price of the option. Since self-exciting processes have quite different paths with respect to Poisson jumps, a priori it would be expected to see large differences in the prices between the models. However, we show next that this is not the case, at least not for the scenarios considered, and that the difference in prices is mostly a vertical shift in the implied volatility surfaces.

For that purpose, we perform the following pricing exercise. We compute the prices of at-the-money Bermudan put options with one-year maturity for different numbers of exercise dates. Naturally, with just one exercise date, we recover the European put price, while, as the number of exercise opportunities increases, the price converges to that of an American option, which is path-dependent. The parameters are set as in Table \ref{tab:parameters}. We present the results in Figures \ref{fig:bermudan_A} and \ref{fig:bermudan_B} for Scenarios \textit{A} and \textit{B}, respectively. 

It is clear that the qualitative behaviour of all models is similar. The implied volatility increases significantly during the first exercise dates, and then it quickly reaches a plateau from where it barely moves. The size of this increment appears to be similar in all models, although in Figure \ref{fig:bermudan_B} it is clear that the Bates model gives rise to the largest increment.

\begin{figure}[ht]
    \centering
    \begin{subfigure}{0.4\textwidth}
    \centering
    \includegraphics[width=\textwidth]{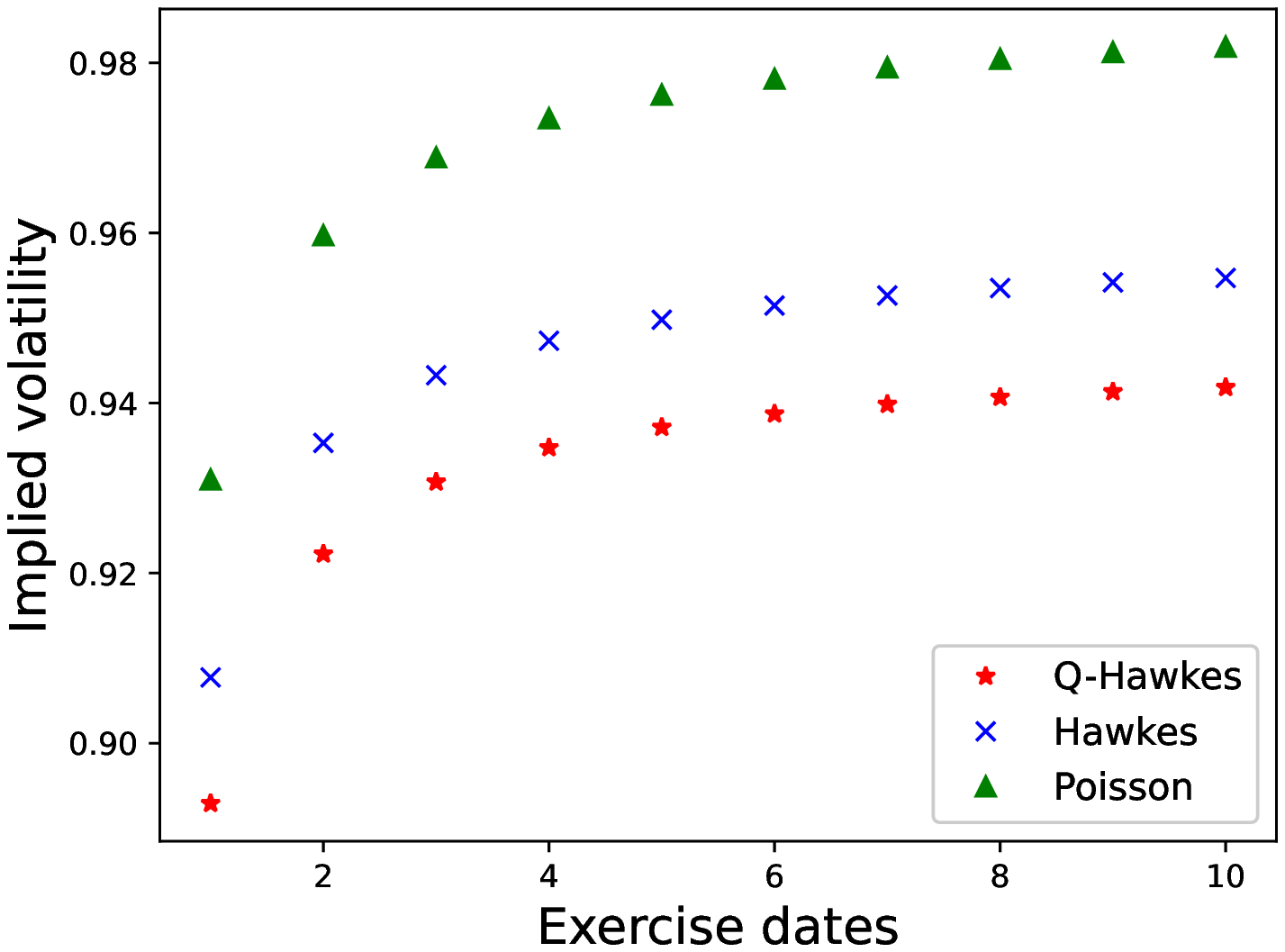}
    \caption{Scenario \textit{A}.}
    \label{fig:bermudan_A}
    \end{subfigure}
    \begin{subfigure}{0.4\textwidth}
    \centering
    \includegraphics[width=\textwidth]{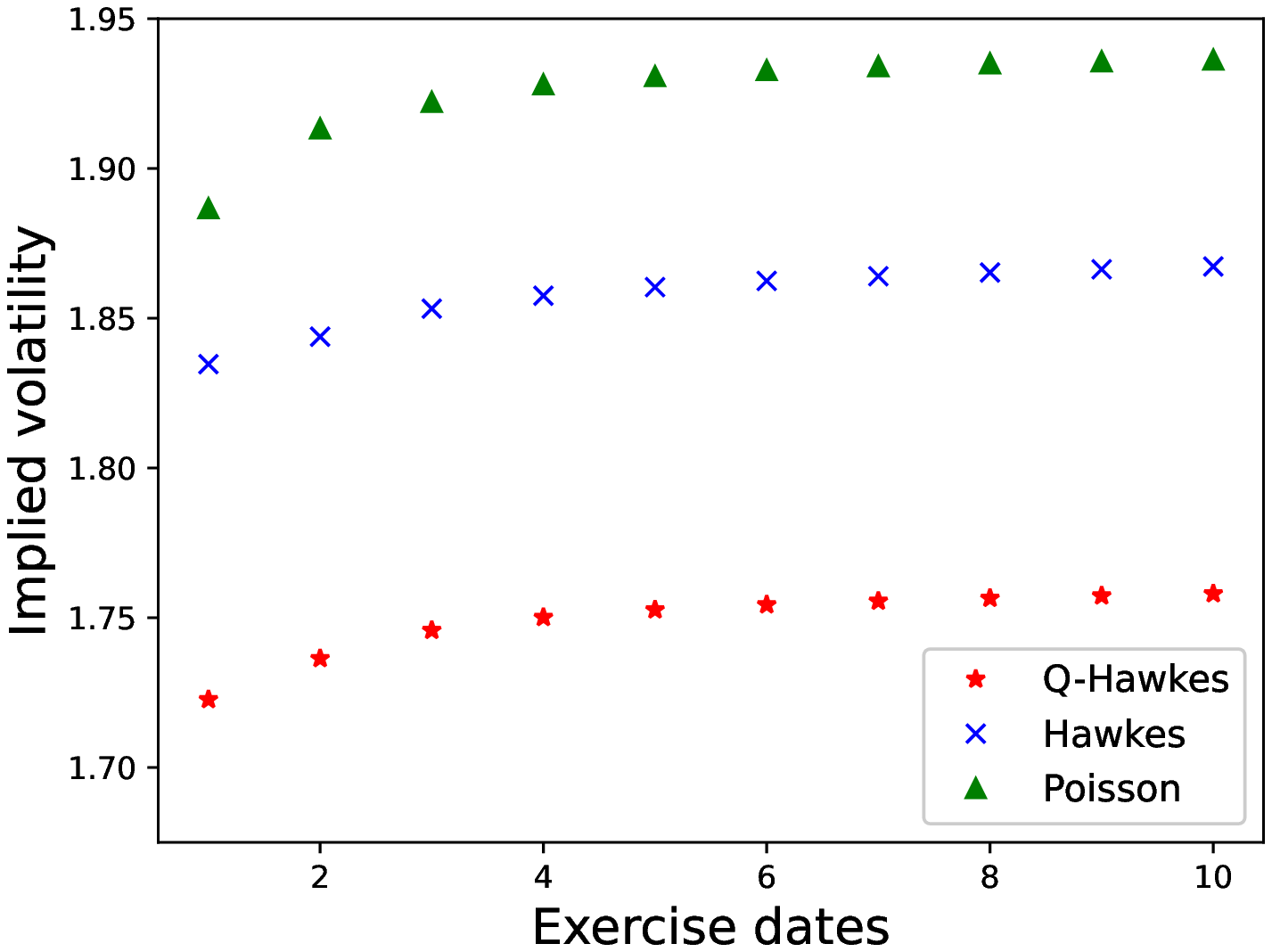}
    \caption{Scenario \textit{B}.}
    \label{fig:bermudan_B}
    \end{subfigure}
    \caption{Implied volatility as a function of the number of exercise dates obtained from Bermudan put options for the HQH, HH and Bates models under Scenarios \textit{A} and \textit{B}. The strike price is chosen so that the option is at-the-money and the maturity is set to one year.}
    \label{fig:bermudan}
\end{figure}

\subsection{Sensitivity analysis}

We finally perform a ceteris-paribus sensitivity analysis for some\footnote{Here we only present the parameters with the largest influence on the shape of volatility smiles and the differences between the HQH and HH models. However, the code used for all the experiments, including the ones not shown here, can be found in the following GitHub repository: \url{https://github.com/LuisSouto/Jumps}.} of the parameters influencing the self-exciting behaviour of the models on the prices of European put options. We compare the differences in prices between the HQH and HH models, and analyze which parameters have the largest impact on those differences\footnote{Since the Bates model is not sensitive to some of these parameters, we do not consider it for this study.}. Furthermore, during the remainder of this section, we assume European put options to be at-the-money with one-year maturity. The parameters are as per Table \ref{tab:parameters}, unless stated otherwise.

As a first experiment, we fix all parameters except the clustering rate $\alpha$. In order to maintain the stability condition, we allow it to vary only within the interval $[0,\beta)$. The results are shown in Figures \ref{fig:sens_a_A} and \ref{fig:sens_a_B}. As expected, the implied volatility increases with the clustering rate, since it adds uncertainty to the model. The shape of the curves is very similar for both models, even overlapping most of the time. For low clustering rates, they both converge to the Bates model, due to the lack of self-excitation. Notwithstanding, for large values of $\alpha$ the differences between the HQH and HH models clearly increase. Moreover, this gap gets more pronounced in Scenario \textit{B} than in Scenario \textit{A}.

\begin{figure}[ht]
    \centering
    \begin{subfigure}{0.4\textwidth}
    \centering
    \includegraphics[width=\textwidth]{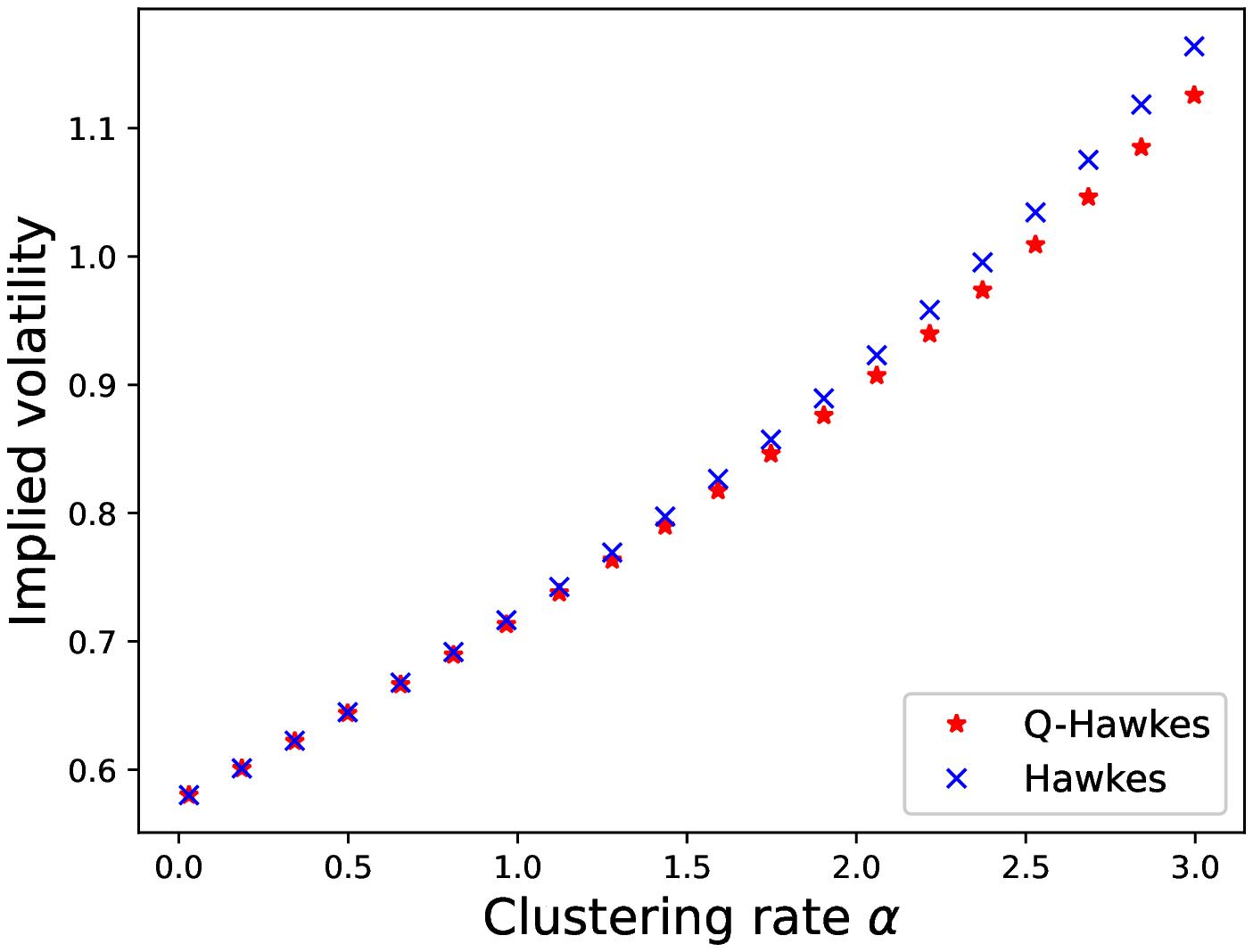}
    \caption{Scenario \textit{A}.}
    \label{fig:sens_a_A}
    \end{subfigure}
    \begin{subfigure}{0.4\textwidth}
    \centering
    \includegraphics[width=\textwidth]{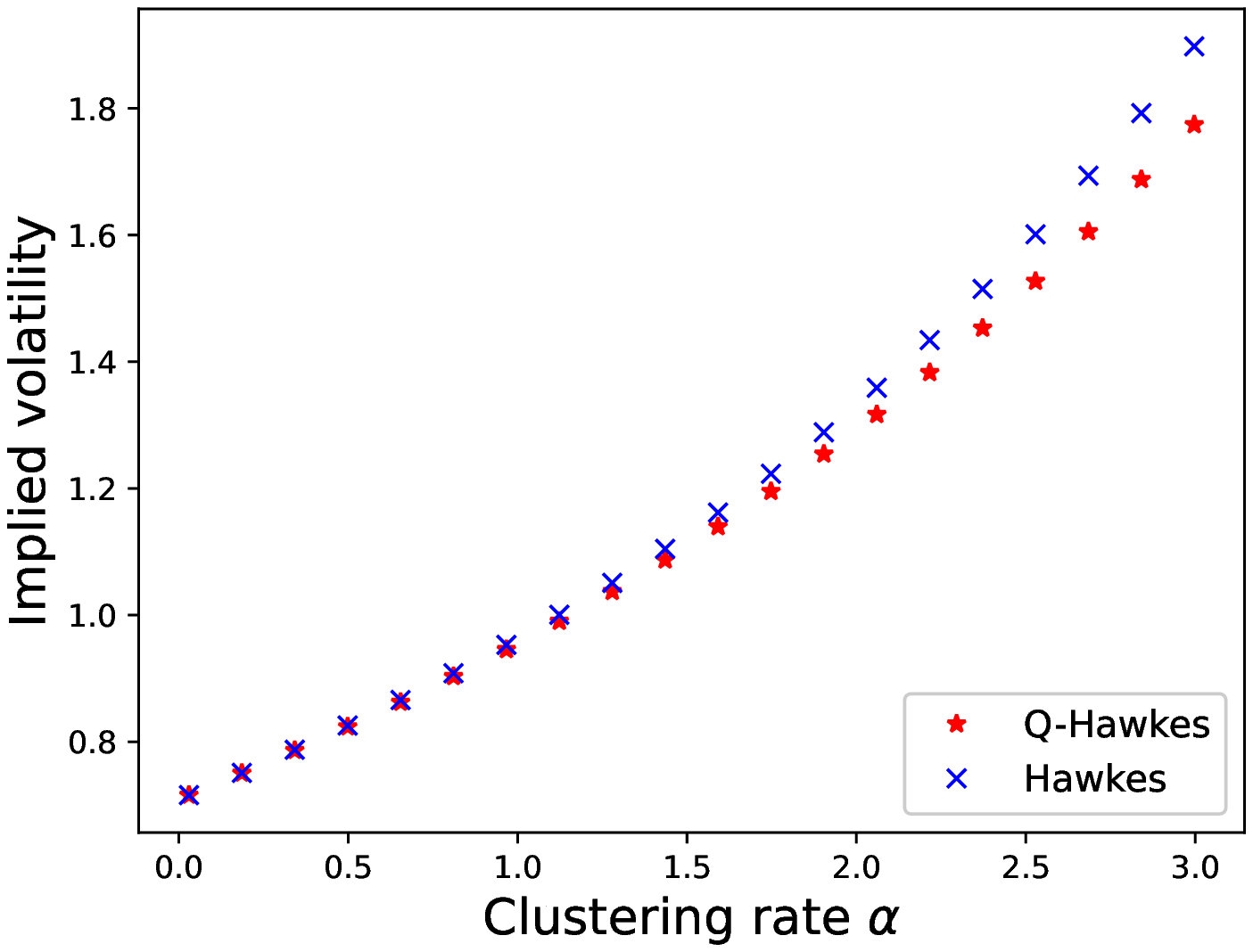}
    \caption{Scenario \textit{B}.}
    \label{fig:sens_a_B}
    \end{subfigure}
    \caption{Implied volatility as a function of the clustering rate obtained from European put options for the HQH and HH models under Scenarios \textit{A} and \textit{B}. The strike price is chosen so that the option is at-the-money and the maturity is set to one year.}
    \label{fig:sens_a}
\end{figure}

For the next experiment, we perform a similar study for the expiration rate $\beta$. Also in order to maintain the stability condition, $\beta$ is only allowed to take values in $(\alpha,\infty)$. Notice that large values of $\beta$ imply an extremely short memory of each event. In the limit $\beta \rightarrow \infty$, the effect of the jumps is immediately eliminated, so there is no self-excitation and the models converge to the Bates model with intensity $\lambda^*$. This can be seen in Figures \ref{fig:sens_b_A} and \ref{fig:sens_b_B}, where the implied volatilities of the HQH and HH models match for large values of $\beta$. The decrease in volatility with $\beta$ is also explained by the corresponding decrease in the uncertainty the models can generate. On the other hand, the discrepancies between the models grow for $\beta$ close to $\alpha$. From this and the previous experiment, we can conclude that the differences between the HQH and HH models increase with the degree of self-excitation, represented by the ratio $\alpha/\beta$.

\begin{figure}[ht]
    \centering
    \begin{subfigure}{0.4\textwidth}
    \centering
    \includegraphics[width=\textwidth]{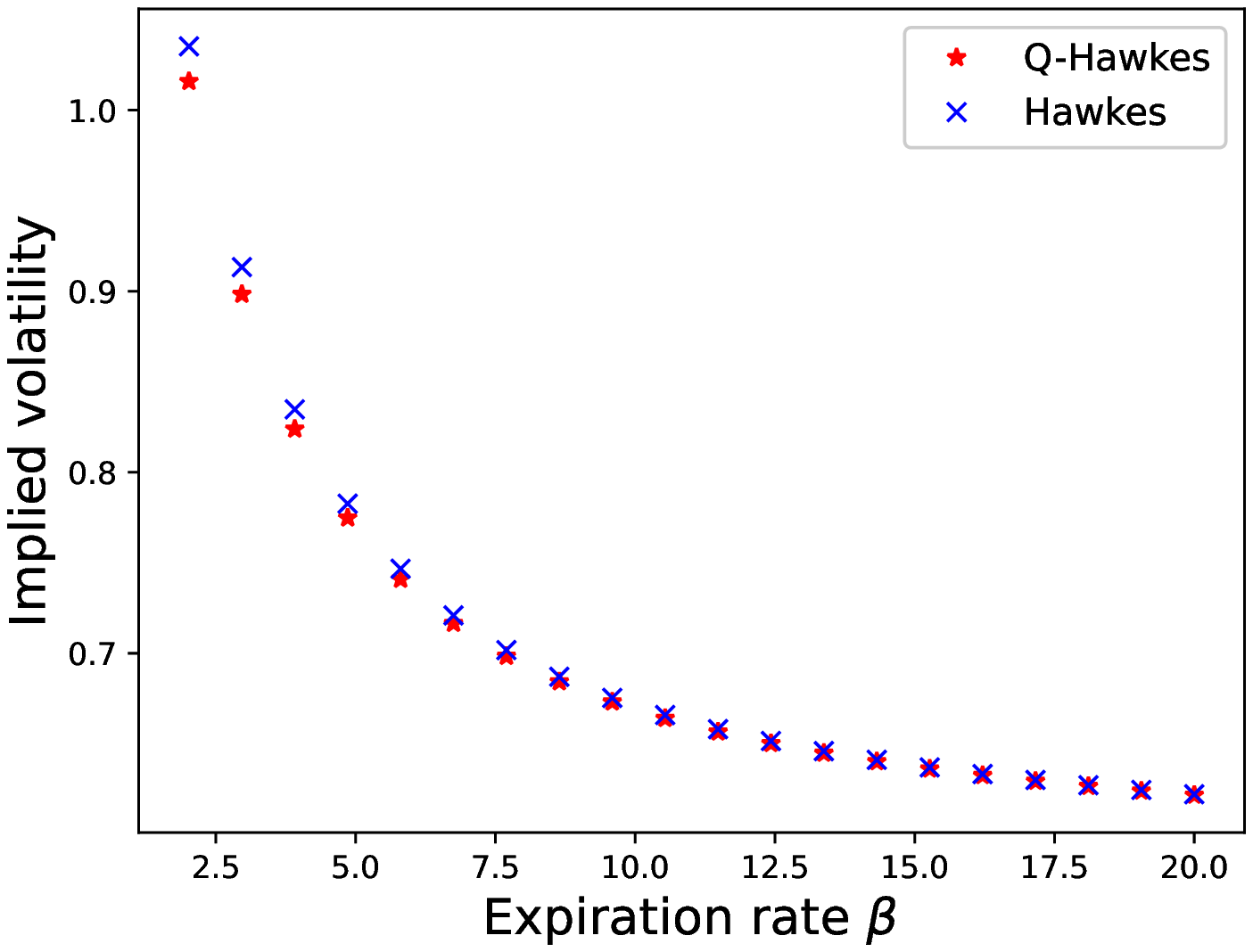}
    \caption{Scenario \textit{A}.}
    \label{fig:sens_b_A}
    \end{subfigure}
    \begin{subfigure}{0.4\textwidth}
    \centering
    \includegraphics[width=\textwidth]{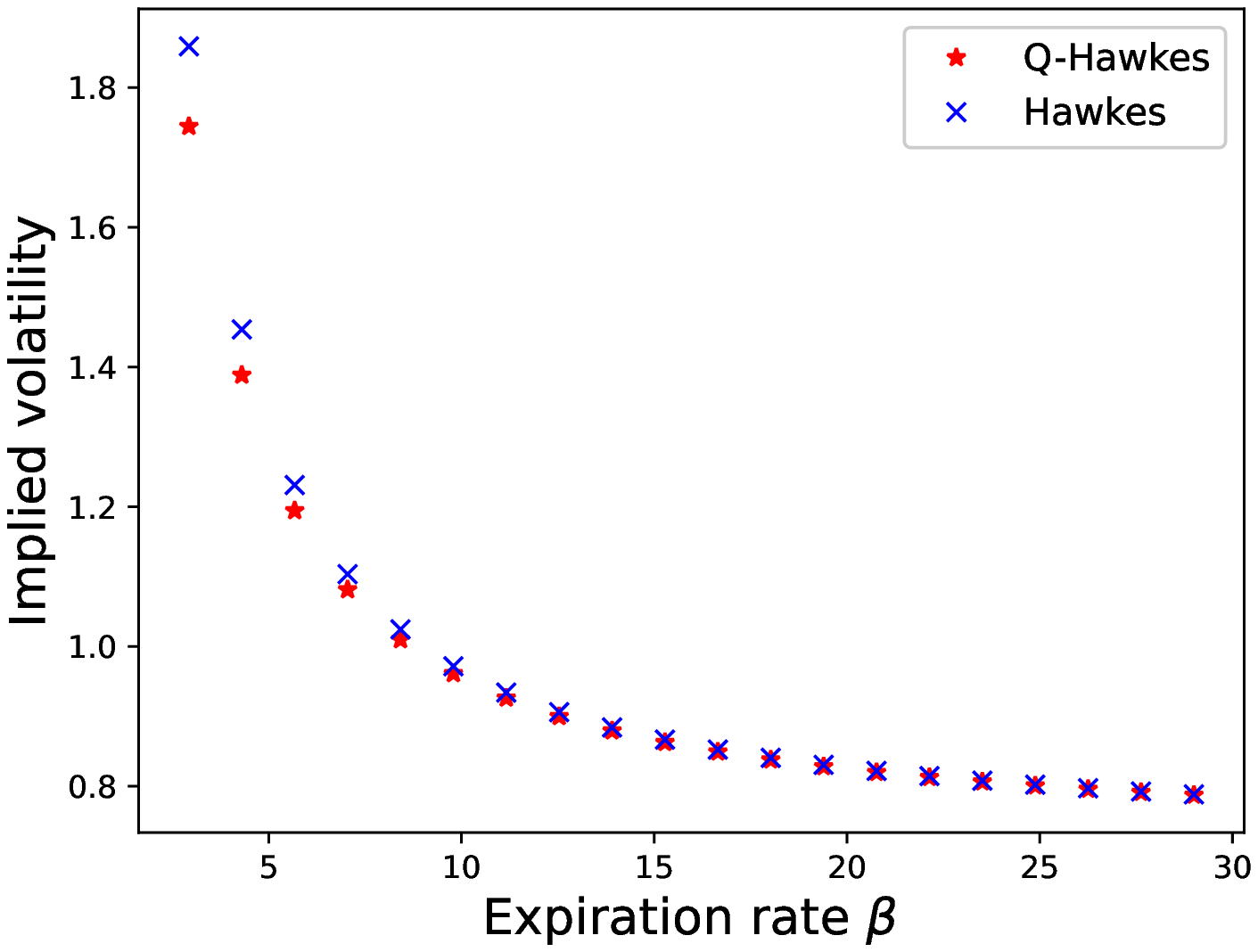}
    \caption{Scenario \textit{B}.}
    \label{fig:sens_b_B}
    \end{subfigure}
    \caption{Implied volatility as a function of the expiration rate obtained from European put options for the HQH and HH models under Scenarios \textit{A} and \textit{B}. The strike price is chosen so that the option is at-the-money and the maturity is set to one year.}
    \label{fig:sens_b}
\end{figure}

The next parameter we analyze is the initial value $Q_0$. The impact of this parameter on the volatility smiles can be seen in Figures \ref{fig:sens_Q0_A} and \ref{fig:sens_Q0_B}. For values of $Q_0$ close to zero, the volatility curves almost overlap, but they move apart as $Q_0$ grows. These effects are much more accentuated in scenario \textit{B}, where the degree of self-excitation is also larger. Notice that a large $Q_0$---for these values of $\alpha$---implies a substantial deviation of $\lambda_0$ from the baseline intensity. Because $\lambda^*$ is the minimum value of the intensity, it can be associated with a state of no self-excitation. Conversely, a large $Q_0$ value means the process is going through a phase of self-excitation. Thus, the ratio $\alpha Q_0/\lambda_0$ gives us an indicator of the state of self-excitation the process is currently in. As with the degree of self-excitation, the differences between the HQH and HH models increase with the state of self-excitation.

\begin{figure}[ht]
    \centering
    \begin{subfigure}{0.4\textwidth}
    \centering
    \includegraphics[width=\textwidth]{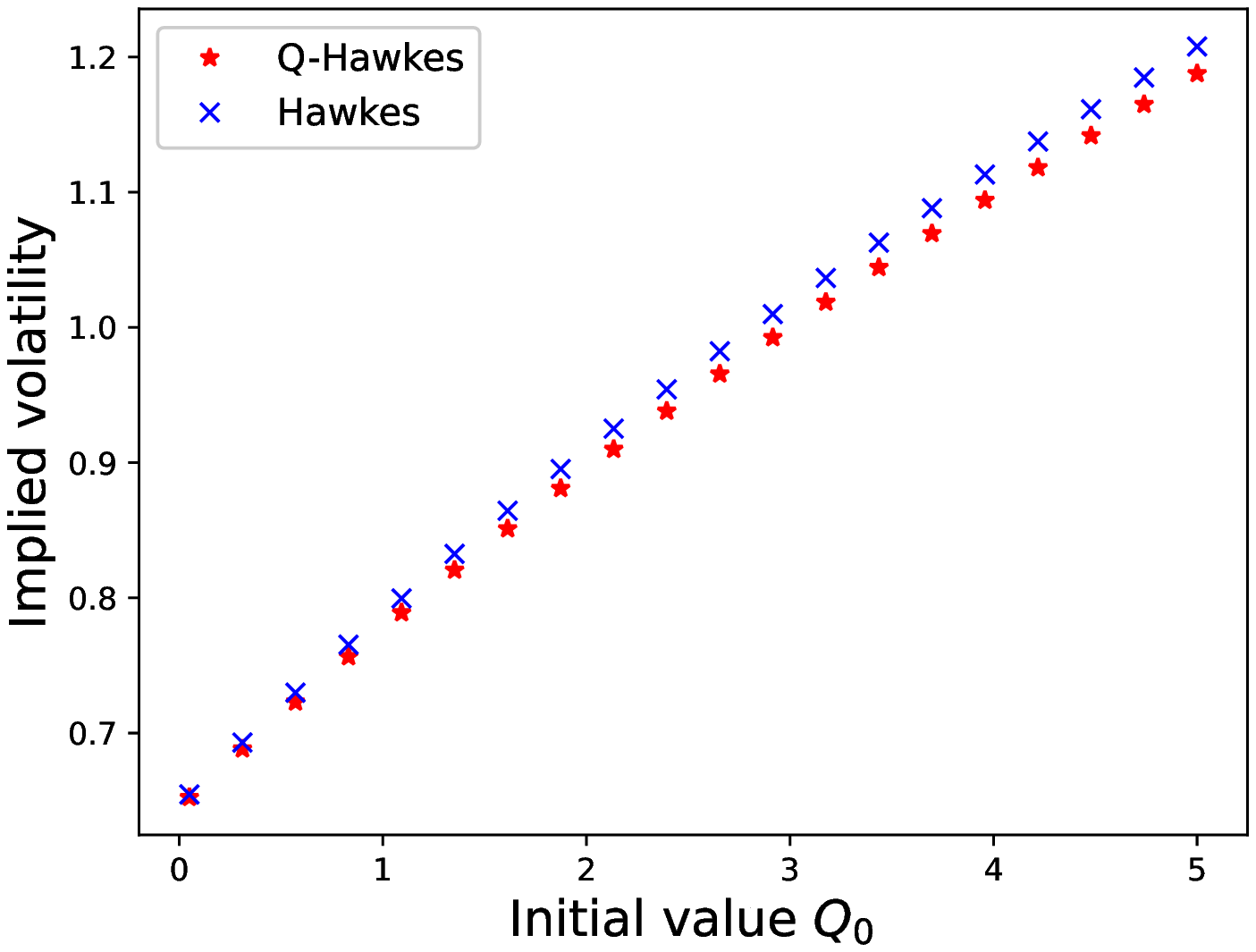}
    \caption{Scenario \textit{A}.}
    \label{fig:sens_Q0_A}
    \end{subfigure}
    \begin{subfigure}{0.4\textwidth}
    \centering
    \includegraphics[width=\textwidth]{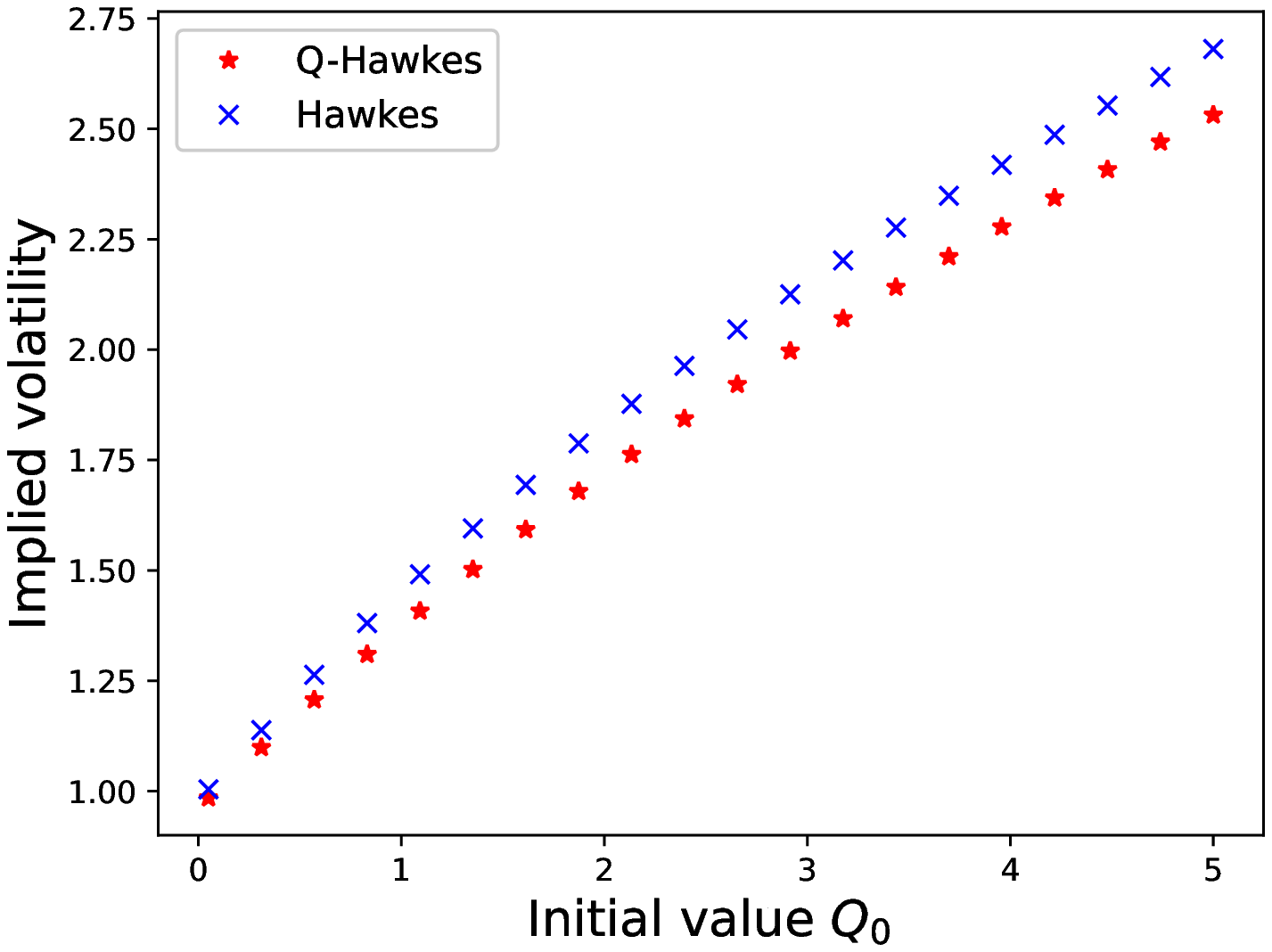}
    \caption{Scenario \textit{B}.}
    \label{fig:sens_Q0_B}
    \end{subfigure}
    \caption{Implied volatility as a function of the initial value $Q_0$, obtained from European put options for the HQH and HH models under Scenarios \textit{A} and \textit{B}. The strike price is chosen so that the option is at-the-money and the maturity is set to one year.}
    \label{fig:sens_Q0}
\end{figure}

The parameters analyzed so far had an influence on the distribution of the number of jumps. Now we study what happens when we modify the distribution of the jump sizes. With $Z$ denoting the jump size, remember that here we assume a normal distribution for $Y = \log(1+Z)$. Thus the parameters are the mean $\mu_Y$ and standard deviation $\sigma_Y$ of $Y$. Since experiments not presented here showed not significant findings about the impact of $\sigma_Y$ on the volatility smiles, next we show only the results for $\mu_Y$.

In Figures \ref{fig:sens_muY_A} and \ref{fig:sens_muY_B} the volatility curves as a function of $\mu_Y$ can be observed. From there we see that this is the only variable in the study which does not produce a monotonic volatility curve. The reason why the volatility grows with large $\mu_Y$ is clear. It increases the expectation of the compensated jump $M_t$ (see Equation \eqref{eq:mt}), which in turn increases the variability of the asset-price. In order to understand the behaviour in the other direction, we first notice that the expectation of $M_t$, in both models, is equal to
\begin{equation}\label{eq:mean_mt}
    \mathbb{E}[M_t-M_s|\mathcal{F}_s] = \frac{\bar{\mu}_Y-\mu_Y}{\beta-\alpha}\left(\lambda^*\beta \left(\frac{e^{(\alpha-\beta)\tau}-1}{\beta-\alpha}-\tau\right)+\lambda_0(e^{(\alpha-\beta)\tau}-1)\right).
\end{equation}
The most important aspect---for our current purposes---about Equation \eqref{eq:mean_mt}, is that it depends on the quantity $(\bar{\mu}_Y-\mu_Y)$. This quantity is not monotonic with $\mu_Y$. Starting from large negative values of $\mu_Y$, it decreases until it reaches a minimum, then it starts growing again. This explains why the volatility curves are not monotonic as functions of this parameter, although there are of course moments of higher order involved. In particular, the minimum of Equation \eqref{eq:mean_mt} does not match with the minimum observed in Figures \ref{fig:sens_muY_A} and \ref{fig:sens_muY_B}. Interestingly, the same conclusions apply to the discrepancies between the HQH and HH models.

\begin{figure}[ht]
    \centering
    \begin{subfigure}{0.4\textwidth}
    \centering
    \includegraphics[width=\textwidth]{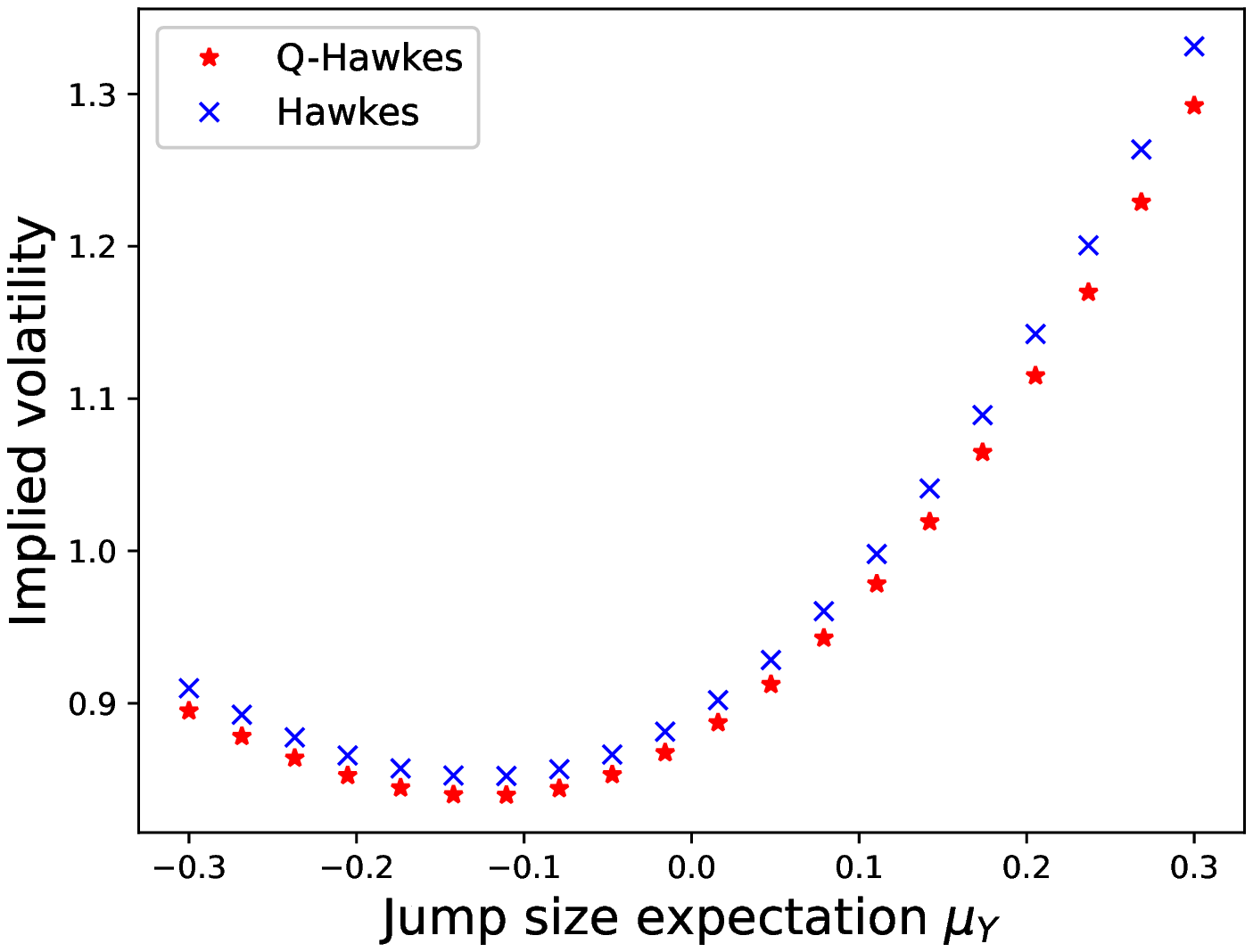}
    \caption{Scenario \textit{A}.}
    \label{fig:sens_muY_A}
    \end{subfigure}
    \begin{subfigure}{0.4\textwidth}
    \centering
    \includegraphics[width=\textwidth]{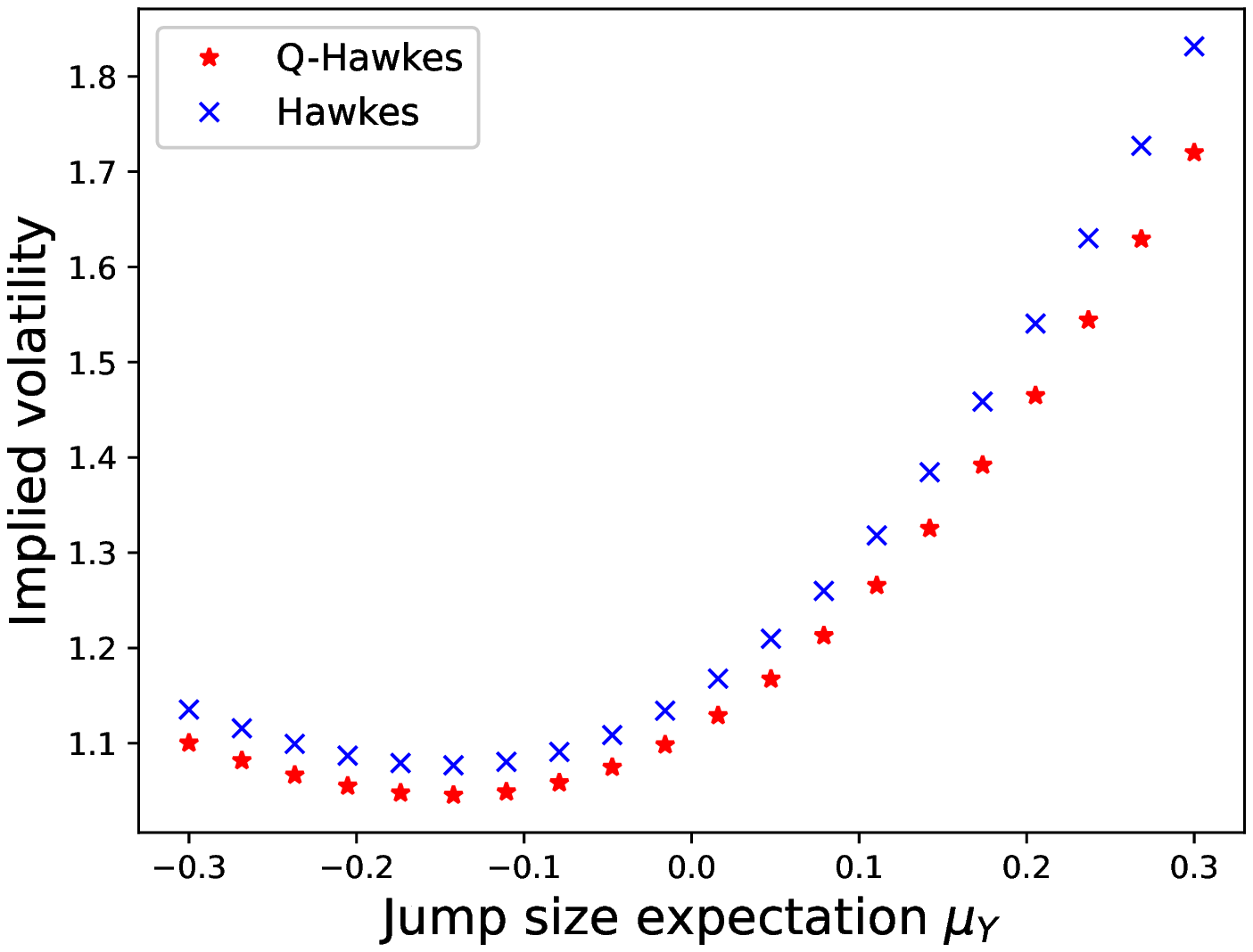}
    \caption{Scenario \textit{B}.}
    \label{fig:sens_muY_B}
    \end{subfigure}
    \caption{Implied volatility as a function of the jump size expectation $\mu_Y$, obtained from European put options for the HQH and HH models under Scenarios \textit{A} and \textit{B}. The strike price is chosen so that the option is at-the-money and the maturity is set to one year.}
    \label{fig:sens_muY}
\end{figure}

\section{Conclusions}\label{sec:conclusions}

We have introduced a novel jump-diffusion process, the HQH model, based on the Q-Hawkes jump process of \cite{Daw2022}. This process enjoys desirable analytical properties, such as a closed-form formula for the characteristic function. Due to this feature, pricing derivatives via Fourier-based algorithms, like the COS method, is computationally cheaper compared to other models, like the HH model. Moreover, we have shown that the COS method can benefit from further properties of the characteristic function, namely partial integrals with respect to some of its arguments. This reduces the dimensionality of the cosine expansion and thus also the numerical complexity. In particular, for the HQH model, the dimensionality can be reduced from three to one dimensions. On the other hand, the cosine expansion of the HH model is two-dimensional.

Numerical experiments indicate that the qualitative behaviour of the implied volatility curves with respect to the option features and model parameters is very similar for the HQH and the HH models. Therefore, the presence of self-excitation in both models makes them more flexible than the traditional Bates model, and for many applications more realistic. Quantitatively, the HH model always gives higher volatilities, although the difference can be considerably small depending on the particular setting.

Considering the promising results of the Q-Hawkes process as a possible tool in the repertoire of jump models in financial applications, our lines of future research will aim at exploring its properties, and proposing generalizations. 

In particular, the Q-Hawkes counterpart of the multidimensional Hawkes process would be of great interest, since this would allow for the inclusion of cross-excitation. Moving to higher dimensions could imply that the characteristic function is no longer analytic, but the process should still be an affine jump-diffusion model. Thus it should have, at least, the same tractability as the Hawkes process. Other non-trivial generalizations are the inclusion of marks and a more general expiration mechanism, so it can also be related to other memory kernels of the Hawkes process that are not exponential.

As already observed in \cite{Daw2022}, an alternative representation in terms of combinatorial stochastic processes--namely urns \cite{mahmoud08}--could also represent a viable and interesting line a work.

\section*{Acknowledgments}
This research has been financed by the European Union, under the H2020-EU.1.3.1. MSCA-ITN-2018 scheme, Grant 813261.

\bibliography{references.bib}
\bibliographystyle{myplainnat}  

\appendix

\section{Complex integral of Proposition \ref{prop:qt_dens}}\label{app:comp_int}

Here we present the details of how to compute the inverse Fourier transform of the characteristic function of a negative binomial distribution. In particular, the aim is to solve the following complex-valued integral
\begin{equation}\label{eq:comp_integral1}
    I(x) \coloneqq \int_{-\infty}^{+\infty} \left( \frac{p(\tau)}{1-(1-p(\tau))e^{i u}} \right)^{\frac{\lambda_0}{\alpha}}
    e^{-i u x} du.
\end{equation}
First, we need to define a contour of integration. In order to do so, notice that the denominator in Equation \eqref{eq:comp_integral1} becomes zero at $u^* = i \log(1-p(\tau))$, and so this is a branching point that we need to take into account when defining the contour. Since\footnote{It is trivial to check that $p(t)$ in Equation \eqref{eq:ese_charq} verifies this condition for any $t \in \mathbb{R}^+$, as long as $\beta > \alpha$. This is the same as the stability condition for the Hawkes process (see Section \ref{sec:hawk}), which we assume throughout this article.} $0\leq p(\tau) \leq 1$, the branching point is negative and purely imaginary. Therefore, and appropriate contour of integration is given in Figure \ref{fig:contour_integral1} (see \cite{Gogolin2014}).
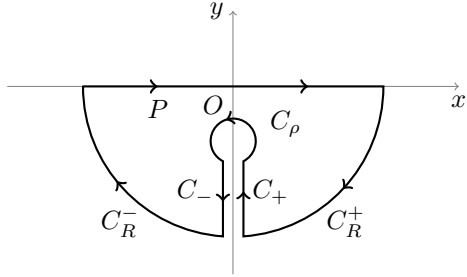
\begin{figure}[ht]
\centering
\begin{tikzpicture}
[decoration={markings,
mark=at position 1cm with {\arrow[line width=1pt]{>}},
mark=at position 3cm with {\arrow[line width=1pt]{>}},
mark=at position 5.5cm with {\arrow[line width=1pt]{>}},
mark=at position 7.6cm with {\arrow[line width=1pt]{>}},
mark=at position 8.9cm with {\arrow[line width=1pt]{>}},
mark=at position 10.14cm with {\arrow[line width=1pt]{>}},
mark=at position 12.24cm with {\arrow[line width=1pt]{>}}
}]
\draw[help lines,->] (-3,0) -- (3,0) coordinate (xaxis);
\draw[help lines,->] (0,-2.5) -- (0,1) coordinate (yaxis);

\path[draw,line width=0.8pt,postaction=decorate] (-2,0) -- (2,0) arc(0:-86:2) -- +(0,1) arc(-63:243:0.3) -- +(0,-1) arc(-94:-180:2);

\node[below] at (xaxis) {$x$};
\node[left] at (yaxis) {$y$};
\node[below left] {$O$};
\node at (-1,-0.3) {$P$};
\node at (0.5,-1.4) {$C_{+}$};
\node at (-0.5,-1.4) {$C_{-}$};
\node at (-1.5,-1.8) {$C_{R}^{-}$};
\node at (1.5,-1.8) {$C_{R}^{+}$};
\node at (0.7,-0.5) {$C_{\rho}$};
\end{tikzpicture}    
    \caption{Integration domain for the integral in Equation \eqref{eq:comp_integral1}. Note that, when integrating, we take limits $\rho \rightarrow 0$ for the radius of the inner arc, and $R \rightarrow \infty$ for the radius of the outer arc. The arc $C_{\rho}$ is centered around the divergence point $u^* = i \log(1-p)$, so that this point lies outside the contour.}
    \label{fig:contour_integral1}
\end{figure}

The next step is to apply Cauchy's theorem to the contour integral. Since the contour of integration is defined outside the divergence point, we obtain
\begin{equation}\label{eq:cauchy_th}
    \int_P + \int_{C_R^+} + \int_{C_+} + \int_{C_{\rho}} + \int_{C_-} + \int_{C_R^-}= 0,
\end{equation}
where  

\begin{equation}
    I(x) = \lim \limits_{\substack{%
    R \to +\infty\\
    \rho \to 0}} \int_P \coloneqq \lim \limits_{\substack{%
    R \to +\infty\\
    \rho \to 0}} \int_P \left( \frac{p(\tau)}{1-(1-p(\tau))e^{i u}} \right)^{\frac{\lambda_0}{\alpha}}
    e^{-i u x} du.
\end{equation}

The remaining part of the proof is to solve these integrals. For that purpose, we assume, without loss of generality\footnote{Notice that we can write $\lambda_0 = \tilde{\lambda}_0 + n \alpha$, where $n$ is an integer and $\tilde{\lambda}_0 < \alpha$. Then the same result follows by applying the derivative operator in Equation \eqref{eq:comp_integral1} and the properties of the Fourier transform.}, that $\lambda_0 < \alpha$. Then the integrals over $C_{\rho}$, $C_R^+$ and $C_R^-$ vanish in the limit $\rho \rightarrow 0$ and $R \rightarrow \infty$, respectively. Moreover, after applying these limits, the integral over the contour $C_+$ is given by:
\begin{equation}
    \int_{C+} \coloneqq \int_{-i\infty}^{u^*} \left( \frac{p(\tau)}{1-(1-p(\tau))e^{i u}} \right)^{\frac{\lambda_0}{\alpha}}
    e^{-i u x} du.
\end{equation}
Applying the change of variables $z = 1-(1-p(\tau))e^{iu}$, we obtain
\begin{equation}\label{eq:int_c+}
    \int_{C_+} = p(\tau)^{\frac{\lambda_0}{\alpha}}(1-p(\tau))^x\int_{-\infty}^0 \frac{z^{-\frac{\lambda_0}{\alpha}}}{(1-z)^{1+x}} \, dz.
\end{equation}
Due to the symmetry of the contour of integration, the integral over $C_-$ is also given by Equation \eqref{eq:int_c+}, with the difference that the complex variable has rotated to a new branch from $C_+$. Therefore, the argument in the integral of $C_-$ is $ze^{i2\pi}$ instead of $z$. Combining both integral gives us
\begin{equation}
    \int_{C_+} + \int_{C_-} = -2p(\tau)^{\frac{\lambda_0}{\alpha}}(1-p(\tau))^x (-1)^{\frac{\lambda_0}{\alpha}}\sin \left(\frac{\pi \lambda_0}{\alpha}\right)\int_{-\infty}^0 \frac{z^{-\frac{\lambda_0}{\alpha}}}{(1-z)^{1+x}} \, dz.
\end{equation}
Next, we perform the change of variables $(1-z)^{-1} = s$ and use Equation \eqref{eq:cauchy_th} to obtain $I(x)$ in terms of $\int_{C_+}$ and $\int_{C_-}$:
\begin{equation}
    I(x) = p(\tau)^{\frac{\lambda_0}{\alpha}}(1-p(\tau))^{x}\sin\left(\frac{\pi \lambda_0}{\alpha}\right)B\left(x+\frac{\lambda_0}{\alpha},1-\frac{\lambda_0}{\alpha}\right),
\end{equation}
where $B(\cdot,\cdot)$ is the Beta function. This expression can be simplified by noticing the following property of the Gamma function:
\begin{equation}
    \Gamma(1-z)\Gamma(z) = \frac{\pi}{\sin(\pi z)},
\end{equation}
which leads to the solution of Equation \eqref{eq:comp_integral1} in terms of the negative binomial distribution.

\end{document}